\documentclass[journal]{IEEEtran}

\usepackage[pdftex]{graphicx}
\usepackage{amssymb}
\usepackage{amsmath}
\usepackage{amsthm}
\usepackage{balance}
\usepackage{cite}
\usepackage{ctable}
\usepackage{subfig}

\title{Spatially-Coupled LDPC Codes for Decode-and-Forward Relaying of Two Correlated Sources over the BEC}
\author{}

\newtheorem{theorem}{Theorem}
\newtheorem{lemma}{Lemma}
\newtheorem{corollary}{Corollary}

\theoremstyle{definition}

\newcommand{\src}{\mathrm{s}}

\newcommand{\rel}{\mathrm{r}}

\newcommand{\dst}{\mathrm{d}}

\newcommand{\kr}{k_\rel}
\newcommand{\krone}{k_\rel^1}
\newcommand{\krtwo}{k_\rel^2}
\newcommand{\kri}{k_\rel^i}
\newcommand{\kt}{\bar{k}}

\newcommand{\none}{n_1}
\newcommand{\ntwo}{n_2}
\renewcommand{\ni}{n_i}
\newcommand{\nr}{n_\rel}
\newcommand{\np}{\tilde{n}}
\newcommand{\N}{N}

\newcommand{\thetai}{\theta_j}

\newcommand{\x}{\mathbf{x}}
\newcommand{\xone}{\x^{(1)}}
\newcommand{\xtwo}{\x^{(2)}}
\renewcommand{\xi}{\x^{(i)}}
\renewcommand{\u}{\mathbf{u}}
\newcommand{\ui}{\u^{(i)}}
\newcommand{\s}{\mathbf{s}}

\renewcommand{\H}{\mathbf{H}}
\newcommand{\synd}{\text{synd}}
\newcommand{\sys}{\text{sys}}
\newcommand{\Hz}{\H_\mathcal{Z}}
\newcommand{\Hsynd}{\H_\synd}
\newcommand{\Hsyndexp}{\left[\begin{array}{cc}\Hsynd^1 & \Hsynd^2\end{array}\right]}

\newcommand{\Rone}{R_1}
\newcommand{\Rtwo}{R_2}
\newcommand{\Ri}{R_i}
\newcommand{\Rr}{R_\rel}

\newcommand{\Ronep}{\tilde{R}_1}
\newcommand{\Rtwop}{\tilde{R}_2}
\newcommand{\Rip}{\tilde{R}_i}

\newcommand{\Roneeff}{R'_1}
\newcommand{\Rtwoeff}{R'_2}

\newcommand{\Reff}{R'}
\newcommand{\Rieff}{R'_i}
\newcommand{\Rsys}{R_{\text{sum}}}

\newcommand{\bl}{\text{bl}}

\newcommand{\Rbli}{R_\bl^i}

\newcommand{\Rblcpone}{\tilde{R}_{\bl,c}^1}
\newcommand{\Rblcptwo}{\tilde{R}_{\bl,c}^2}
\newcommand{\Rblcip}{\tilde{R}_{\bl,c}^i}
\newcommand{\Rbl}{R_\bl}

\newcommand{\Rsone}{R_\synd^1}
\newcommand{\Rstwo}{R_\synd^2}
\newcommand{\Rsi}{R_\synd^i}

\newcommand{\Rsrcone}{R_{\src1}}
\newcommand{\Rsrctwo}{R_{\src2}}
\newcommand{\Rsrci}{R_{\src i}}

\newcommand{\Rcone}{R_{\mathrm{c}1}}
\newcommand{\Rctwo}{R_{\mathrm{c}2}}
\newcommand{\Rci}{R_{\mathrm{c}i}}
\newcommand{\Rcip}{\tilde{R}_{\mathrm{c}i}}
\newcommand{\Rcpone}{\tilde{R}_{\mathrm{c}1}}
\newcommand{\Rcptwo}{\tilde{R}_{\mathrm{c}2}}

\newcommand{\Rconeeff}{R'_{\mathrm{c}1}}
\newcommand{\Rctwoeff}{R'_{\mathrm{c}2}}
\newcommand{\Rcieff}{R'_{\mathrm{c}i}}

\newcommand{\nCheck}{N_{\text{C}}}

\newcommand{\nCheckS}{N_{\text{C,synd}}}

\newcommand{\nVar}{N_{\text{V}}}

\newcommand{\epsoner}{\epsilon_{\src_1\rel}}
\newcommand{\epstwor}{\epsilon_{\src_2\rel}}
\newcommand{\epsir}{\epsilon_{\src_i\rel}}
\newcommand{\epsr}{\epsilon_{\src\rel}}
\newcommand{\epsoned}{\epsilon_{\src_1\dst}}
\newcommand{\epstwod}{\epsilon_{\src_2\dst}}
\newcommand{\epsid}{\epsilon_{\src_i\dst}}
\newcommand{\epsd}{\epsilon_{\src\dst}}
\newcommand{\eprd}{\epsilon_{\rel\dst}}

\newcommand{\mapThresh}{\epsilon^{\text{MAP}}}
\newcommand{\bpThresh}{\epsilon^{\text{BP}}}
\newcommand{\shanLim}{\epsilon^{\text{Sh}}}

\newcommand{\Csoner}{C_{\src_1\rel}}
\newcommand{\Cstwor}{C_{\src_2\rel}}
\newcommand{\Csir}{C_{\src_i\rel}}
\newcommand{\Csr}{C_{\src\rel}}
\newcommand{\Csoned}{C_{\src_1\dst}}
\newcommand{\Cstwod}{C_{\src_2\dst}}
\newcommand{\Csid}{C_{\src_i\dst}}
\newcommand{\Csd}{C_{\src\dst}}
\newcommand{\Crd}{C_{\rel\dst}}

\newcommand{\Ccode}{\mathcal{C}}
\newcommand{\Cone}{\mathcal{C}_1}
\newcommand{\Ctwo}{\mathcal{C}_2}
\newcommand{\Ci}{\mathcal{C}_i}
\newcommand{\Csone}{\mathcal{C}_\synd^1}
\newcommand{\Cstwo}{\mathcal{C}_\synd^2}
\newcommand{\Csi}{\mathcal{C}_\synd^i}
\newcommand{\Cbl}{\mathcal{C}_\text{bl}}

\newcommand{\Mone}{M_1}
\newcommand{\Mtwo}{M_2}
\newcommand{\Mi}{M_i}

\newcommand{\lone}{l_1}
\newcommand{\ltwo}{l_2}
\newcommand{\rone}{r_1}
\newcommand{\rtwo}{r_2}

\newcommand{\lsone}{l_\synd^1}
\newcommand{\lstwo}{l_\synd^2}
\newcommand{\lsi}{l_\synd^i}
\newcommand{\rsone}{r_\synd^1}
\newcommand{\rstwo}{r_\synd^2}
\newcommand{\rsi}{r_\synd^i}
\newcommand{\ls}{l_\synd}
\newcommand{\rs}{r_\synd}

\newcommand{\corr}{\text{corr}}

\newcommand{\pit}[3]{p_{#1}^{(#2,#3)}}
\newcommand{\pits}[3]{p_{\synd,#1}^{(#2,#3)}}
\newcommand{\qit}[3]{q_{#1}^{(#2,#3)}}
\newcommand{\qits}[3]{q_{\synd,#1}^{(#2,#3)}}

\newcommand{\pitc}[3]{p_{\corr,#1}^{(#2,#3)}}

\newcommand{\otoprule}{\midrule[\heavyrulewidth]}

\newcommand{\nCheckSEff}{N_{\text{C,synd}}^{\text{eff}}}

\newcommand{\T}{^{\mathsf{T}}}

\begin{document}

\author{Stefan~Schwandter,
        Alexandre~Graell~i~Amat,~\IEEEmembership{Senior Member,~IEEE,}
        and~Gerald~Matz,~\IEEEmembership{Senior~Member,~IEEE}
\thanks{S. Schwandter is with Grabner Instruments, Ametek Inc., Dr.-Otto-Neurath-Gasse 1, A-1220 Vienna, Austria (e-mail: s.schwandter@me.com).}
\thanks{A. Graell i Amat is with the Department of Signals and Systems, Chalmers University of Technology, SE-412 96 Gothenburg, Sweden (e-mail: alexandre.graell@chalmers.se).}
\thanks{G. Matz is with the Institute of Telecommunications, Vienna University of Technology, Gusshausstrasse 25/389, A-1040 Vienna, Austria (e-mail: gmatz@nt.tuwien.ac.at).}
\thanks{Part of this work has been previously presented at the IEEE 7th International Symposium on Turbo Codes \& Iterative Information Processing (ISTC) 2012.
}%
\thanks{The work of S.~Schwandter and G.~Matz was funded by FWF Grant S10606 and WWTF Grant ICT08-44. A.~Graell i Amat was supported by the Swedish
Research Council under grant \#2011-5961, and by the Swedish Foundation for Strategic Research (SSF) under the Gustaf Dal\'en project IMF11-0077.}
}

\maketitle
\begin{abstract}
We present a decode-and-forward transmission scheme based on spatially-coupled low-density parity-check (SC-LDPC) codes for a network consisting of two (possibly correlated) sources, one relay, and one destination. The links between the nodes are modeled as binary erasure channels. Joint source-channel coding with joint channel decoding is used to exploit the correlation. The relay performs network coding. We derive analytical bounds on the achievable rates for the binary erasure time-division multiple-access relay channel with correlated sources. We then design bilayer SC-LDPC codes and analyze their asymptotic performance for this scenario. We prove analytically that the proposed coding scheme achieves the theoretical limit for symmetric channel conditions and uncorrelated sources. Using density evolution, we furthermore demonstrate that our scheme approaches  the theoretical limit also for non-symmetric channel conditions and when the sources are correlated, and we observe the threshold saturation effect that is typical for spatially-coupled systems. Finally, we give simulation results for large block lengths, which validate the DE analysis.
\end{abstract}

\section{Introduction}

The three-node relay channel was introduced by van der Meulen in \cite{van1971three} and the first capacity results were presented in \cite{Cover79}. Recent years have seen a vast amount of research on relaying, both in the information theory and coding communities. While the capacity for the general relay channel is still unknown, a number of relaying strategies have been devised, which establish achievable rates. 
One of the most prominent examples is the \emph{decode-and-forward} (DF) relaying scheme, introduced in \cite{Cover79}. 
With DF, the relay fully decodes the source data 
and provides a re-encoded copy of the source message to the destination. Apart from information theoretical considerations, there has also been a huge interest in designing practical relaying schemes. 
Several papers have considered practical implementations of DF based on convolutional codes \cite{Stefanov04}, turbo codes \cite{Zhao03} or low-density parity-check (LDPC) codes 
\cite{Chakrabarti07}. In \cite{Razaghi2007}, so-called \emph{bilayer} (BL-) LDPC codes were introduced and were shown to closely approach the theoretical DF rate.
Recently, a combination of the bilayer structure with spatially-coupled LDPC (SC-LDPC) codes \cite{felstrom1999,Lentmaier10,Kud11} was investigated in \cite{Si2011} and it was shown that BL-SC-LDPC codes can actually achieve the Shannon limit of a DF relay system with orthogonal binary erasure channel (BEC) links. As the SC-LDPC code ensembles are regular, the design complexity is very low compared to schemes based on irregular LDPC code ensembles, which require extensive optimization.

Despite the complexity that is inherent already in the three-node relay channel, more complicated networks have been investigated as well. For example, in a practical system, the need may arise that one relay assists more than just one single source.
Such a system, in which multiple sources share one relay, is modeled by the multiple access relay channel (MARC). Capacity results for the MARC with independent sources were given in \cite{kramer2004information}. Information theoretical bounds for the MARC with correlated sources have been recently given in\cite{Murin2011}.
A common assumption is that the transmissions in the relay system are orthogonalized using time division multiple-access (TDMA). For the time division MARC (TD-MARC), several coding schemes  have been proposed based on regular LDPC codes \cite{hausliterative}, irregular LDPC codes \cite{linetwork}, and serially-concatenated codes \cite{youssef2011distributed}. The authors recently proposed a scheme based on SC-LPDC codes \cite{schwandter12}. However, none of these works considered correlated sources.

The task of efficiently transmitting the correlated data from two or more source nodes to one or more destination nodes in a communications network is a topic currently undergoing high research activity. This so-called ``sensor reach-back problem'' \cite{Barros:2002fk,barros2006network} occurs, e.g., in wireless sensor networks \cite{akyildiz2002survey}, where measurements of neighboring sensors can be spatially correlated. In the simplest case, two correlated sources transmit their data directly to one common destination. For independent discrete memoryless channels, it has been shown that the separation of source and channel coding is asymptotically optimum in this (and more general) scenario(s) \cite{Barros:2002fk,barros2006network}. Therefore, the achievable rates for such a system can be derived assuming this separation.
However, for designing practical schemes, many works take a joint source-channel coding (JSCC) approach, where the uncompressed source data is directly encoded with channel codes, and the correlation is exploited at the receiver using joint channel decoding (JCD). There are several reasons for not separating source and channel coding in a practical system: First, the design of practical source codes for correlated sources is an open problem \cite{zhao2003lossless}. Secondly, errors introduced by the channel decoder could be catastrophic for the source decoder. Thirdly, the implementation of JCD in form of an iterative decoder based on a factor graph \cite{Kschischang:2001p6575} is conceptually simple and appealing, and it allows analysis through the density evolution (DE) technique \cite{richardson2008modern}.
%
Early works on practical transmission schemes for two correlated sources were based on turbo codes and low-density generator matrix (LDGM) codes\cite{daneshgaran2006ldpc,zhong2005ldgm,garcia2007turbo}. More recent works, e.g., \cite{abrardo2009optimizing,martalo2010density}, used DE and extrinsic information transfer (EXIT) chart techniques to optimize irregular LDPC codes. Recently, a scheme based on regular SC-LDPC codes has been shown to achieve a performance close to the theoretical limits \cite{yedla2011universality}. The results suggest that this near-optimum performance is due to the threshold saturation effect exhibited by SC-LDPC codes \cite{Kud11}. 




In this paper, we consider a system in which the data of two (possibly correlated) sources is transmitted to a common destination, with the help of a relay. The links between the nodes are modeled as BECs and they are orthogonalized using TDMA.
The main contributions of this paper are the following: We derive an achievable rate for the TD-MARC with correlated sources.
Since the links in our system are independent, the maximum rate is achieved by separate source and channel coding \cite{barros2006network}. We then propose a two-user bilayer SC-LDPC coded relaying scheme for this scenario and analyze its behavior for the case of asymptotically large block length. The system uses JSCC for the transmission to both the relay as well as the destination nodes. In addition, the relay implicitly uses network coding to combine the sources data before forwarding it to the destination.
Since the SC-LDPC codes used are regular, their design is simple, does not involve the optimization of the degree distributions, and simplifies to choosing appropriate node degrees of the component codes for given link qualities. The factor-graph-based design of the joint source-channel-network decoder at the destination node incorporates aspects of the decoders for correlated sources described in \cite{yedla2009can} and extends the two-user bilayer relaying scheme proposed by the authors in \cite{schwandter12} to the scenario with correlated sources. We give DE results that show that the performance of the proposed relaying scheme is very close to the theoretical limit. We also show that the phenomenon of threshold saturation occurs, which is responsible for the outstanding performance. Finally, for the special case of uncorrelated sources and symmetric link capacities, we prove that our scheme achieves the maximum rate achievable by decode-and-forward relaying. Simulation results for large block lengths are also given.

The remainder of the paper is structured as follows. Section~\ref{sec:sysmodel} introduces the system model. In Section~\ref{sec:theoretical_limits}, we revisit the achievable DF rates of the two-source TD-MARC and extend the analysis to the case where the sources are correlated. In Section~\ref{sec:scldpccodes}, we give a brief overview of SC-LDPC codes. Our main contribution, the two-user bilayer SC-LDPC coded relaying scheme, is presented in Section~\ref{sec:twoUserBl}. DE for the proposed bilayer SC-LDPC codes is discussed in Section~\ref{sec:decStructDe}. A proof that the proposed coding schemes achieves the DF TD-MARC theoretical limit for symmetric channel conditions and uncorrelated sources is also given in this section. In Section~\ref{sec:achChanParam}, we derive the achievable region of channel parameters, which serves as a theoretical benchmark to which the proposed bilayer codes must be compared. DE results are presented in Section~\ref{sec:numresults}, which demonstrate that our scheme performs very closely to the theoretical limit also for non-symmetric channels and correlated sources, and show the threshold saturation phenomenon. Finally, Section~\ref{sec:conclusions} concludes the paper.

\section{System model}
\label{sec:sysmodel}

\begin{figure}
\begin{center}
\includegraphics[]{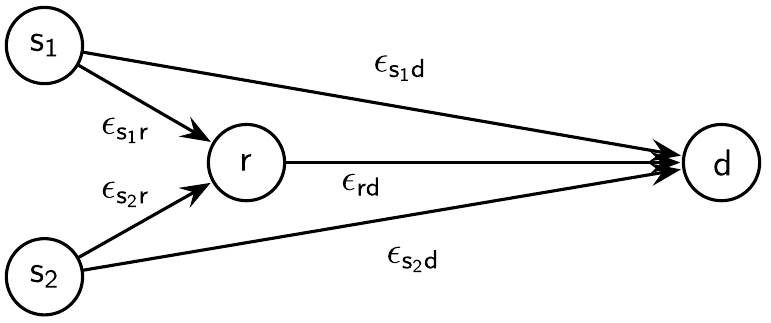}
\caption{Two-user relay network in which all links are modeled as binary erasure channels (links are labeled with erasure probabilities).}
\label{fig:sys_model}
\end{center}
\vspace{-3ex}
\end{figure}
We consider the relay network depicted in Fig.~\ref{fig:sys_model}, consisting of four nodes: Two sources, $\src_1$ and $\src_2$, transmit correlated data to a single destination $\dst$ with the help of a relay $\rel$. The links between the nodes are modeled as BECs with erasure probabilities $\epsoner$, $\epstwor$, $\epsoned$, $\epstwod$, and $\eprd$, where the first subscript denotes the transmitter, and the second subscript the receiver in each link.  One transmission block is split into three phases in order to orthogonalize the transmission links. In each phase only one node transmits and the others listen. 
In phase $i$ ($i=1,2$), source 
$\src_i$ 
transmits a length-$\ni$ codeword $\x^{(i)}\triangleq\big(x_1^{(i)},\ldots,x_{n_i}^{(i)}\big){\T}$, obtained by encoding $k_i$ information bits $\u^{(i)}\triangleq\big(u_1^{(i)},\ldots,u_{k_i}^{(i)}\big){\T}$ 
with a code of rate $\Ri \triangleq k_i/\ni$. 
The relay receives $\x^{(i)}$
over the  $\src_i$-$\rel$ link, which has 
capacity $\Csir=1-\epsir$.
The destination receives the same codewords over
the $\src_i$-$\dst$ links with capacities $\Csid=1-\epsid$. The relay decodes both transmissions, combines the decoded sources data as described below, and encodes the resulting $\kr$ bits into a codeword 
of length $\nr$ using a code of rate $\Rr$. 
In the third phase, this codeword is
forwarded to the destination over the $\rel$-$\dst$ link, with capacity $\Crd=1-\eprd$. The whole transmission block consists of $N\triangleq\none+\ntwo+\nr$ bits. The length of each phase relative to the total block length is $\thetai\triangleq n_j/\N$, $j=(1,2,\rel)$.
The \emph{effective}\footnote{Throughout the paper we use a prime superscript to indicate an effective transmission rate.} transmission rates of the sources are defined as
\begin{equation}
\Rieff\triangleq\frac{k_i}{N} = \frac{k_i \theta_i}{n_i}=\theta_i R_i\quad \text{source bits/transmission block}\label{eq:effvsrealrates}.
\end{equation}
The sum rate of the system is $\Rsys\triangleq \frac{k_1+k_2}{N}=\Roneeff+\Rtwoeff$.



In the system described above, source coding is not considered. For deriving the maximum achievable rate for this system, however, we consider that the sources data is compressed before transmission: Source $\src_i$ compresses $k_i$ information bits with a source coding rate 
\begin{equation*}
\Rsrci \triangleq\frac{\kt_i}{k_i},
\end{equation*}
resulting in $\kt_i$ compressed bits. Then, it uses a channel code of rate $\Rci\triangleq\frac{\kt_i}{n_i}$ to encode the compressed bits into $n_i$ code bits. The overall transmission rates including source and channel coding are therefore
\begin{equation}
\Ri = \frac{k_i}{n_i} = \frac{k_i}{\kt_i}\frac{\kt_i}{n_i}=\frac{\Rci}{\Rsrci} \quad \text{source bits/channel use}.\label{srcCodingRate}
\end{equation}
The effective channel code rates are defined as
\begin{equation}
\Rcieff\triangleq \frac{\kt_i}{N} = \frac{\kt_i \theta_i}{n_i} = \theta_i \Rci\label{eq:effChanCodeRates}.
\end{equation}

Combining \eqref{eq:effvsrealrates}, \eqref{srcCodingRate}, and \eqref{eq:effChanCodeRates} results in
\begin{equation}
\Rieff= \frac{\Rcieff}{\Rsrci}.\label{ReffVsRcieff}
\end{equation}

The correlation between the sources is modeled in the following way \cite{yedla2009can}:
Let $Z_n$ be a binary random variable (RV) in a length-$k$ sequence $\mathbf{Z}$ of independent and identically distributed
RVs with $\Pr(Z_n=1)=p$. The source bits $U_n^{(1)}$ and $U_n^{(2)}$ have the following relationship,
\begin{equation*}
\left(U^{(1)}_n,U^{(2)}_n\right) = \begin{cases}\text{independent Bernoulli-}\frac{1}{2} \text{ RVs, if }Z_n=0\\ 
                         \text{same Bernoulli-}\frac{1}{2} \text{ RV } U_i\text{, if }Z_n=1.
            \end{cases}
\end{equation*}
It is assumed that the decoder knows the realization $\mathbf{z}$ of $\mathbf{Z}$ \cite{yedla2009can}.
Let $\mathbf{U}^{(i)}=(U^{(i)}_1,\ldots,U^{(i)}_k){\T}$. The joint entropy of the two source bit sequences $\mathbf{U}^{(1)}$ and $\mathbf{U}^{(2)}$ is (with slight abuse of notation for brevity) $H(U_1,U_2)=2-p$, and the conditional entropies are $H(U_1|U_2)=H(U_2|U_1)=1-p$.

\section{Theoretical limits}
\label{sec:theoretical_limits}

In this section we derive the achievable DF rates for the TD-MARC described in the previous section, when the sources are correlated. Since the derivations are based on the achievable DF rates for uncorrelated sources, we first revisit those for clarity purposes.

\subsection{Uncorrelated sources}
 
The achievable DF transmission rates for a TD-MARC with independent sources are given by \cite{Hausl2008}
\begin{align}
\Rconeeff &\le \theta_1 \Csoner \label{eq:sr1} \\
\Rctwoeff &\le \theta_2 \Cstwor \label{eq:sr2} \\
\Rconeeff &\le \theta_1 \Csoned + \theta_r \Crd \label{eq:srd1}\\
\Rctwoeff &\le \theta_2 \Cstwod + \theta_r \Crd \label{eq:srd2}\\
\Rconeeff + \Rctwoeff &\le \theta_1 \Csoned + \theta_2 \Cstwod + \theta_r\Crd\label{eq:sumConst}.
\end{align}
Under the assumptions $\Csoner\ge\Csoned$, $\Cstwor\ge\Cstwod$, $\Crd\ge\Csoned$ and $\Crd\ge\Cstwod$, 
the time allocation leading to the maximum achievable effective rates is \cite{Hausl2008}
\begin{align}
\label{eq:theta1opt}
\theta_1^* &= \frac{\Crd}{(1+\sigma \kappa) \Crd + (1+ \sigma) \Csoner- \Csoned - \sigma\kappa \Cstwod}\\
\label{eq:theta2opt}
\theta_2^* &= \kappa \theta_1^*, \\
\theta_r^* &=1-\theta_1^*-\theta_2^*,
\end{align}
where $\kappa \triangleq \Csoner/\Cstwor$ and $\sigma\triangleq\Rctwoeff/\Rconeeff$. 
This time allocation results in the effective rates
$
\Rconeeff=\frac{\Csoner\Crd}{(1+\sigma \kappa) \Crd + (1+ \sigma) \Csoner- \Csoned - \sigma\kappa \Cstwod}
$ 
and
$
\Rctwoeff=\sigma\Rconeeff
$. 
Note that in the uncorrelated case, the effective channel coding rates are equal to the effective transmission rates, $\Rieff=\Rcieff$.


\subsection{Correlated sources}
\label{thlim_corr}
For the transmission of correlated sources to a common destination over independent discrete memoryless channels, it was shown in \cite{barros2006network} that it is optimum to consider source coding and channel coding separately (for infinitely long blocks): First, the source data is compressed up to the Slepian-Wolf limit \cite{Slepian73}, then the compressed data is channel encoded with capacity-achieving codes and transmitted over the channels. %
We apply this strategy to obtain the achievable rates for the transmission of correlated sources over the TD-MARC. 
The achievable source coding rates given by Slepian and Wolf \cite{Slepian73} are
\begin{align}
\Rsrcone &\ge H(U_1|U_2) \label{eq:sw1}\\
\Rsrctwo &\ge H(U_2|U_1) \label{eq:sw2}\\
\Rsrcone + \Rsrctwo &\ge H(U_1,U_2) \label{eq:sw3},
\end{align}

The part of the region given by
\eqref{eq:sw1}-\eqref{eq:sw3}
that gives the lowest possible rates (i.e., maximal compression of the sources) is the one given by
\begin{equation}
H(U_1|U_2) \le \Rsrcone \le H(U_1) \label{eq:rs2ofrs1},\quad
\Rsrctwo(\Rsrcone) = H(U_1,U_2) - \Rsrcone.
\end{equation}
We will implicitly assume \eqref{eq:rs2ofrs1} whenever we write $\Rsrctwo$ from now on.

In the rest of the paper, we assume $k_1=k_2\triangleq k$, which entails $\Reff_1=\Reff_2\triangleq\Reff$, for simplicity. From the bounds on the achievable effective transmission rates for the TD-MARC with uncorrelated sources \eqref{eq:sr1}-\eqref{eq:sumConst}, considering the fact that we compress the sources data before transmission over the links \eqref{ReffVsRcieff}, we obtain the following bounds for the achievable effective transmission rate for the TD-MARC with correlated sources,
\begin{align}
\Reff& \le \frac{1}{\Rsrcone}\theta_1 \Csoner   \triangleq f_1\label{eq:sr1x}\\
\Reff& \le \frac{1}{\Rsrctwo} \theta_2 \Cstwor  \triangleq f_2\label{eq:sr2x} \\
\Reff& \le \frac{1}{\Rsrcone} \left(\theta_1 \Csoned + \theta_r \Crd \right)  \triangleq f_3\label{eq:srd1x}\\
\Reff& \le \frac{1}{\Rsrctwo}\left( \theta_2 \Cstwod + \theta_r \Crd \right)
	\triangleq f_4 \label{eq:srd2x}\\
\Reff& \le \frac{1}{H(U_1,U_2)} \left(\theta_1 \Csoned + \theta_2 \Cstwod + \theta_r\Crd\right)   \triangleq f_5\label{eq:sumConstx},
\end{align}
where the source coding rates $\Rsrci$ (with the restrictions in 
\eqref{eq:rs2ofrs1}) and the time allocation parameters $\theta_1$, $\theta_2$ and $\theta_r=1-\theta_1 - \theta_2$ are free parameters that have to be optimized to obtain the maximum achievable rate
\begin{equation*}
\Reff_{\text{max}}\triangleq\max_{\theta_1,\theta_2,\Rsrcone} \Reff.
\end{equation*}
This maximum 
is given in the following theorem.

\begin{theorem} The maximum rate achievable on the TD-MARC with correlated sources is
\begin{equation*}
R_{\text{max}}' = \frac{1}{\Rsrcone^*}\theta_1^* \Csoner,
\end{equation*}
where the optimum time allocation parameters and source coding rates are
\begin{align*}
\theta_1^* &= \frac{\Crd}{\left(1+\kappa'\right) \Crd + \frac{H(U_1,U_2)}{\Rsrcone^*}\Csoner - \Csoned - \kappa'\Cstwod}\label{eq:theta1_opt_thm} \\
\theta_2^* &= \kappa'\theta_1^*\\
\Rsrcone^* &=
\begin{cases}
H(U_1|U_2) &\quad\text{if } \kappa > \nu \\
1 &\quad\text{if } \kappa < \nu \\
\text{arbitrary in }\, [H(U_1|U_2),1] &\quad\text{if } \kappa = \nu,
\end{cases}
\end{align*}
with
\begin{equation*}
\kappa=\frac{\Csoner}{\Cstwor},\quad \kappa'=\kappa\left(\frac{H(U_1,U_2)}{\Rsrcone^*}-1\right), \quad \nu=\frac{\Crd-\Csoned}{\Crd-\Cstwod}.
\end{equation*}
\end{theorem}
\begin{proof}
See Appendix.
\end{proof}

\section{SC-LDPC codes}
\label{sec:scldpccodes}


We briefly review SC-LDPC codes. A regular $(l,r)$ SC-LDPC code with variable node degree $l$ and check node degree $r$ is defined by an infinite parity-check matrix
\begin{equation}
\H^{\T}=
\left[\begin{matrix} 
 \ddots &   & \ddots &   &   \\
& \hspace*{-2em}\H^{\T}_0(0) & \hspace*{-2em}\dots & \hspace*{-2em}\H^{\T}_{m_s}(m_s) &  &  \\  
 & \qquad\ddots &  & \qquad\ddots &  \\  
&   & \hspace*{-1em}\H^{\T}_0(t) & \hspace*{-1em}\ldots & \hspace*{-1em}\H^{\T}_{m_s}(t+m_s)  \\  
 &   & \qquad\ddots &   & \qquad\ddots
\end{matrix}\right],
\label{eq:ldpcc_matrix}
\end{equation}
%
where superscript $\T$ denotes the matrix transpose.
The Tanner graph describing such a code is divided into ``positions'' or ``time instants'' $t$, similar to the trellis sections in classical convolutional codes. At each position $t \in (-\infty,\infty)$ there are $M$ variable nodes, and $M\frac{l}{r}$ check nodes. This is reflected in the parity-check matrix by the fact that each submatrix $\H^{\T}_j(t+j)$, $j\in[0,m_s]$, is a sparse $M \times (M\frac{l}{r})$ binary matrix.
For our application we will consider \emph{terminated} spatially-coupled code ensembles, where the codeword is restricted to 
$t\in [1,L]$, and the parity-check matrix is therefore of finite size.

We use the regular 
$(l,r,L,w,M)$ ensemble described in \cite{Kud11}. In this ensemble, a variable node at position $t$ has $l$ connections to check nodes at positions from the range $[t,t+w-1]$,
where $w=m_s+1$. For each connection, the position of the check node is uniformly and independently chosen from that range. This randomization results in simple DE equations and thus renders the ensemble accessible to analysis. For transmission over the BEC, the code rate of the $(l,r,L,w,M)$ ensemble tends to the one of the underlying block code ensemble,
$
\lim_{w\rightarrow \infty}\lim_{L\rightarrow \infty}\lim_{M\rightarrow \infty} R(l,r,L,w,M) = 1- \frac{l}{r}
$, 
as $M$, $L$ and $w$ go to infinity, in that order.
Furthermore, its belief propagation (BP) threshold $\bpThresh$ 
tends to the maximum a posteriori (MAP) threshold $\epsilon^{\text{MAP}}$ 
of the underlying ensemble, i.e.,
$
\lim_{w\rightarrow \infty}\lim_{L\rightarrow \infty}\lim_{M\rightarrow \infty} \bpThresh(l,r,L,w,M)
=\mapThresh(l,r).
$ 
On the BEC, the Shannon limit for transmission at rate $R$ is given by $\shanLim=1-R$, below which 
reliable (error-free) transmission is possible.
The MAP threshold of a regular LDPC block code ensemble tends to the Shannon limit exponentially fast in $l$ if the design rate $R(l,r)$ is kept fixed (cf. \cite[Lemma 8]{Kud11}),
\begin{equation}
\lim_{l\rightarrow\infty}\mapThresh\left(l, r = \frac{l}{1-R(l,r)}\right)= 1-R (l,r) = \shanLim.
\label{eq:capAch}
\end{equation}
Therefore, the BP threshold of the $(l,r,L,w,M)$ ensemble asymptotically tends to the Shannon limit.

\section{Two-user bilayer SC-LDPC code}
\label{sec:twoUserBl}

We now present the main contribution of this paper, a two-user bilayer SC-LDPC code for the two-source relaying scenario with correlated sources. Before we go into the details of the relaying scheme, we first give a high-level overview of the proposed coding structure: Each source uses an SC-LDPC code to transmit its data. The rates of the codes are chosen such that the decoder at the relay can perfectly recover the sources data (note that, since the decoder exploits the known correlation between the sources, the rates can be higher than the respective link capacities). The relay recovers the codewords transmitted by the sources, and generates additional ``syndrome'' bits, which are transmitted to the destination protected by a channel code that allows error-free transmission. The decoder at the destination uses the known correlation of the sources and, additionally, the syndrome bits provided by the relay, to construct an overall code for both users that is decoded jointly in order to reconstruct the original sources data.

\subsection{Code structure}
\label{sec:twoUserBlCodeStructure}
Sources $\src_1$ and $\src_2$ use codes from the ensembles $\Cone(\lone,\rone,L,w,\Mone)$ and $\Ctwo(\ltwo,\rtwo,L,w,\Mtwo)$, respectively, with design rates $\Ri=\frac{k}{\ni}=\frac{k}{L\Mi}$, and parity-check matrices $\H^1\in [0,1]^{(\none-k)\times n_1}$ and $\H^2\in [0,1]^{(\ntwo-k)\times n_2}$ (for simplicity, we assume $L_1=L_2=L$ and $w_1=w_2=w$ without loss of generality).
They constitute the first layer of the bilayer structure.
Their rates are designed such that the relay is able to decode the source data error-free (see Section \ref{sec:rateDesign}). If the codes in the first layer are designed properly, the relay can recover the codewords $\xone$ and $\xtwo$ transmitted by the sources. It then generates $\kr$ additional syndrome bits according to
\begin{equation}
\s = \Hsyndexp \left[\begin{array}{c}\xone \\\xtwo \end{array}\right ],
\label{eq:synd_bits}
\end{equation}
using parity-check matrices $\H_\synd^1\in [0,1]^{\kr\times n_1}$ and $\H_\synd^2\in [0,1]^{\kr\times n_2}$ from the code ensembles $\Csone(\lsone,\rsone,L,w,\Mone)$ and $\Cstwo(\lstwo,\rstwo,L,w,\Mtwo)$, respectively. The extra parity-checks introduced by $\H_\synd^1$ and $\H_\synd$ constitute the second layer of the bilayer code.

The syndrome bits are transmitted to the destination protected by another SC-LDPC code of rate $\Rr=\Crd$ (its design is independent of the other codes and will therefore not be considered further). It is assumed that this code can be decoded error-free, and separately from the other codes in the system. With that assumption, the destination can now decode the source bits using the 
parity-check matrix $\H$ of the \emph{overall} code
\begin{equation}
\mathbf{H} \left[\begin{array}{c}\xone \\ \xtwo \end{array} \right] = \left[\begin{array}{cc}\mathbf{H}_1 &   \mathbf{0}\\  \mathbf{0}&\mathbf{H}_2 \\ 
\Hsynd^1 &   \Hsynd^2 \\
\H_\corr^1 &  \H_\corr^2\end{array}\right]
\left[\begin{array}{c}\xone \\\xtwo \end{array} \right] = \left[\begin{array}{c}\mathbf{0} \\ \mathbf{0} \\ \s \\ \mathbf{0}   \end{array} \right].
\label{eq:overall_decoding_joint}
\end{equation}
\begin{figure}
\begin{center}
\includegraphics[scale=0.7]{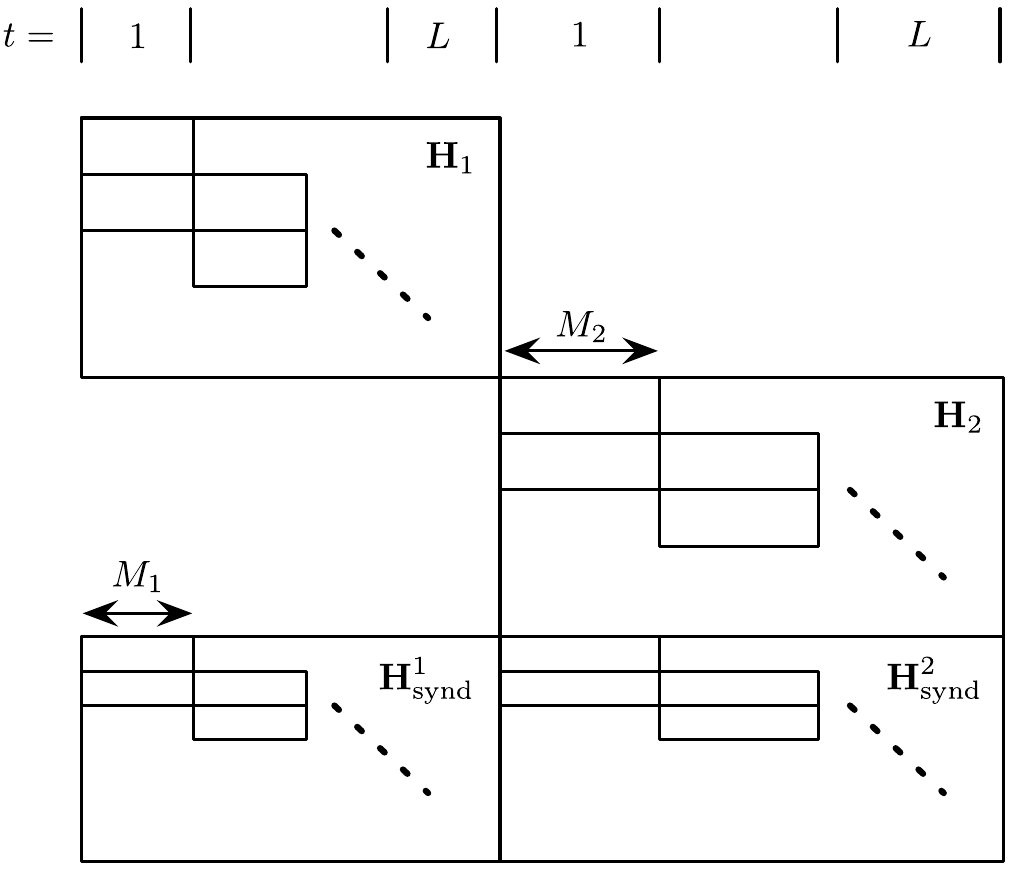}

\caption{Two-user bilayer SC-LDPC code parity-check matrix structure (without the part due to correlation).}
\label{fig:ldpcc_structure}
\end{center}
\vspace{-3ex}
\end{figure}
Here, $\H_\corr^i=\Hz\H_\sys^i$ denotes a $(k \times n_i)$ matrix, which represents the parity-check equations that are due to the correlation model: The $(k\times n_i)$ matrices $\H_\sys^i$ are used to ``extract'' the systematic bits of the corresponding codeword, i.e., $\H_\sys^i\xi=\ui$. $\Hz$ is a  diagonal $(k\times k)$ matrix, where the $n$-th diagonal entry is equal to $z_n$ (the $n$-th entry of $\mathbf{z}$). In this way, the parity-check $u_n^{(1)}\oplus u_n^{(2)}=0$ is enforced for all indices $n$ for which $z_n=1$.

Since the height of each component matrix in $\Hsynd^i$ is $M_i\frac{\lsi}{\rsi}$ (cf. Section~\ref{sec:scldpccodes}), we see that the condition
\begin{equation}
M_1\frac{\lsone}{\rsone}=M_2\frac{\lstwo}{\rstwo}\label{eq:consistencyAss}
\end{equation}
has to hold to have a proper alignment of the entries in $\Hsynd^1$ and $\Hsynd^2$ (cf. Fig.~\ref{fig:ldpcc_structure}). This alignment is necessary for the DE analysis in Section~\ref{sec:decStructDe}.

We note that the overall matrix $\H$ does not 
have the band structure of an SC-LDPC code (cf.~\eqref{eq:ldpcc_matrix}). Nevertheless, we will analytically prove capacity-achieving performance 
for symmetric channel conditions when there is no correlation between the sources and, based on DE results, we conjecture that capacity is achieved under general channel conditions and with source correlation.

From \eqref{eq:synd_bits} we see that each syndrome bit depends in general on code bits from both sources. Of the total number $\kr$ of syndrome bits, effectively $\kri$ bits ($\kr=\krone+\krtwo$) are used to respectively decode $\src_i$ at the destination. The number of bits from source $\src_i$ involved in one of the $\kr$ checks is given by the check node degree $\rsi$, since the degree is equal to the number of ones in a row of $\Hsynd^i$. The effective number of syndrome bits for source $\src_i$ is therefore
$\kri=\kr \mu_i$,
where
\begin{equation}
\mu_i\triangleq\frac{\rsi}{\rsone+\rstwo}
\label{eq:n_synd_bits}
\end{equation}
is the ratio between the number of bits of codeword $\x^{(i)}$ involved in one particular syndrome bit in $\s$ (cf. (\ref{eq:synd_bits})) and the total number of code bits involved in that syndrome bit. We also define
$\mu\triangleq\frac{\mu_1}{\mu_2}=\frac{\rsone}{\rstwo}$ 
for later use.

%
The codes that are used at the relay to generate the syndrome bits have rates
\begin{equation}
\Rsone \triangleq 1-\frac{\kr}{\none}, \qquad
\Rstwo \triangleq 1-\frac{\kr}{\ntwo}\label{eq:Rsynd_def}.
\end{equation}
They are related according to
\begin{equation}
\frac{1-\Rsone}{1-\Rstwo}=\frac{R_1}{R_2}\label{eq:Rsynd1Rsynd2}
\end{equation}
due to \eqref{eq:Rsynd_def} and $\frac{n_1}{n_2}=\frac{R_2}{R_1}$, which follows from \eqref{srcCodingRate} and the assumption $k_1=k_2=k$.


\subsection{Code rate design}
\label{sec:rateDesign}

In the following, we derive the design rates $R_1$, $R_2$, $\Rsone$, and $\Rstwo$ of the component codes of the two-user bilayer code as well as the assignment $\mu_1$ and $\mu_2$ of syndrome bits to the sources, for given link capacities. We distinguish between the cases of uncorrelated sources and correlated sources.

\subsubsection{Uncorrelated sources}

In order to ensure reliable transmission to the relay, we use channel codes of rate 
$R_i=\Csir$ 
at the sources.
The destination has to be able to decode the bilayer code consisting of the bits transmitted by a source plus the additional syndrome bits from the relay. The rate of the overall bilayer code for source $\src_i$ is
\begin{equation}
\Rbli \triangleq\frac{\ni-(\ni-k +\kri)}{\ni}=\frac{k-\kri}{\ni} =R_i-\mu_i(1-\Rsi).
\label{eq:bl_rate_def}
\end{equation}


For reliable transmission, we have to set $\Rbli=\Csid$.
Since $\mu_i=\frac{R_i-\Rbli}{1-\Rsi}=\frac{\Csir-\Csid}{1-\Rsi}$ (from \eqref{eq:bl_rate_def}), we obtain using \eqref{eq:Rsynd1Rsynd2} and $\mu_1+\mu_2=1$,
\begin{equation}
\mu = \frac{\mu_1}{\mu_2}=\frac{\Cstwor}{\Csoner}\frac{\Csoner-\Csoned}{\Cstwor-\Cstwod},\quad
\mu_1=\frac{1}{1+\mu^{-1}}, \quad \mu_2 = \frac{1}{1+\mu}.\label{sec:mu1mu2}
\end{equation}

Therefore, the rates of the codes used at the relay to generate additional syndrome bits have to be set to
\begin{equation}
\Rsi=1-\frac{1}{\mu_i}(R_i-\Rbli)=1-\frac{1}{\mu_i}(\Csir-\Csid)\label{mu1mu2},
\end{equation}
i.e.,
\begin{align}
\Rsone&=1-\left( 2\Csoner-\kappa\Cstwod-\Csoned\right),\nonumber\\
\Rstwo&=1-\left( 2\Cstwor-\frac{1}{\kappa}\Csoned-\Cstwod\right).\nonumber
\end{align}

It can easily be shown that this choice of code rates together with $\Rr=\Crd$ leads to the optimum time allocation \eqref{eq:theta1opt}, 
\begin{align*}
\theta_1&=\frac{n_1}{n_1+n_2+n_r}=\frac{1}{1+\frac{n_2}{n_1}+\frac{k_r}{\Rr n_1}}\\
        &=\frac{\Crd}{\Crd(1+\kappa)+(1-\Rsone)}=\theta_1^*,
\end{align*}
for $\sigma=1$. Similar calculations verify $\theta_2=\theta_2^*$ and $\theta_r=\theta_r^*$.

\subsubsection{Correlated sources}

To derive the rates for the case of correlated sources, we again turn to the system with separate source-channel coding, which we used to find the optimum time allocation in Section \ref{thlim_corr}. The joint source-channel coded system 
then uses channel codes of rates $R_1$ and $R_2$ that are equal to the \emph{transmit} rates (i.e., rates including source and channel coding) of the system with separate source-channel coding.

In the case of correlated sources, the codewords (as opposed to the information sequences) of the two sources should ideally be independent for transmission over independent channels \cite{barros2006network}. Therefore, we consider punctured systematic codes \cite{yedla2009can}, i.e., all systematic bits are punctured, to minimize the correlation between codewords. If all  systematic bits are punctured, only $\np_1 \triangleq n_1-k$ and $\np_2 \triangleq n_2-k$ coded bits are transmitted in the first two phases of the transmission block. The total transmission length is then $N=\np_1+\np_2+n_r$ and we redefine $\theta_1=\np_1/N$ and $\theta_2=\np_2/N$. With these redefinitions, the relation between the effective and the actual transmission rates \eqref{eq:effvsrealrates} becomes
\begin{equation}
\Reff=\frac{k}{N} = \frac{k \theta_i}{n_i-k}=\theta_i \frac{R_i}{1-R_i}\label{eq:effvsrealratesPunct}.
\end{equation}

We define the punctured code and transmission rates
\begin{equation*}
\Rcip\triangleq \frac{\kt_i}{n_i-\kt_i} = \frac{\Rci}{1-\Rci},\quad
\Rip\triangleq \frac{k}{n_i-k} = \frac{\Ri}{1-\Ri}.
\end{equation*}


From \eqref{eq:effvsrealratesPunct}, \eqref{eq:sr1x} and \eqref{eq:sr2x}, we obtain the optimum transmission rates for the first layer,
\begin{equation*}
R_1 = \phi\left(\frac{\Csoner}{\Rsrcone^*}\right)
, \qquad \Rtwo=\phi\left(\frac{\Cstwor}{\Rsrctwo^*}\right),
\end{equation*}
where $\phi(x)\triangleq \frac{x}{1+x}$. In terms of the punctured transmission rates,
\begin{equation}
\tilde{R}_1 = \frac{\Csoner}{\Rsrcone^*}
, \qquad \tilde{R}_2=\frac{\Cstwor}{\Rsrctwo^*}.\label{eq:r1r2corrp}
\end{equation}


The rates of the codes of the second layer can be obtained from the optimum time allocation with $\Rr=\Crd$ and
\begin{equation}
\frac{\theta_r}{\theta_i}=\frac{\nr}{n_i-k}=\frac{\kr}{(n_i-k)\Rr}=\frac{1}{\Rr(1-R_i)}(1-\Rsone)\label{eq:rsynd_1_theta_p}.
\end{equation}
We have
\begin{align}
\Rsone&=1-\Rr (1-R_1) \left( \frac{1}{\theta_1^*}-\kappa'(\Rsrcone^*)-1 \right)\nonumber\\
&=1-(1-R_1)\left( \frac{H(U_1,U_2)}{\Rsrcone^*}\Csoner-\kappa'\Cstwod-\Csoned\right)\nonumber
\end{align}
and
\begin{align}
\Rstwo&=1-\Rr (1-\Rtwo)\left( \frac{1}{\theta_2^*}-\frac{1}{\kappa'(\Rsrcone^*)}-1 \right)\nonumber\\
&=1-(1-R_2)\left( \frac{H(U_1,U_2)}{\Rsrctwo^*}\Cstwor-\frac{1}{\kappa'}\Csoned-\Cstwod\right).\nonumber
\end{align}
The punctured bilayer code rates are
\begin{equation*}
\Rblcip\triangleq\frac{\kt_i-\kri}{n_i-\kt_i}=\frac{1}{1-\Rci}\left( \Rci-\mu_i(1-\Rsi)\right),
\end{equation*}
and they have to satisfy $\Rblcip=\Csid$. Therefore we get
\begin{align*}
\mu&=\frac{\mu_1}{\mu_2}=\frac{\Rcone-\Rblcpone(1-\Rcone)}{\Rctwo-\Rblcptwo(1-\Rctwo)} \cdot \frac{1-\Rstwo}{1-\Rsone}\\
&= \frac{\Rcpone-\Rblcpone}{\Rcptwo-\Rblcptwo} \cdot \frac{1-\Rcone}{1-\Rctwo}\cdot\frac{\Rctwo}{\Rcone}\\
&=\frac{\Rcpone-\Rblcpone}{\Rcptwo-\Rblcptwo} \cdot \frac{\Rcpone}{\Rcptwo}
=\frac{\Csoner-\Csoned}{\Cstwor-\Cstwod} \cdot \frac{\Cstwor}{\Csoner},
\end{align*}
and we can obtain $\mu_1$ and $\mu_2$ from \eqref{sec:mu1mu2}.

With that, we have determined the parameters $R_1$, $R_2$, $\Rsone$, $\Rstwo$, $\mu_1$ and $\mu_2$ of the bilayer SC-LDPC code for the relaying system with two correlated sources for given link capacities.

\section{Density Evolution and a Proof for Symmetric Conditions and Uncorrelated Sources}
\label{sec:decStructDe}

In this section, we give the DE equations for the proposed bilayer SC-LDPC codes on the BEC. First, we treat the case of uncorrelated sources, then we extend our consideration to the case of correlated sources. We also prove that the proposed code construction achieves the maximum decode-and-forward rate for the case of symmetric conditions and uncorrelated sources. 

\subsection{Density evolution for uncorrelated sources}

When there is no correlation between the sources, the optimum decoder at the relay consists of two separate decoders for codes $\Cone$ and $\Ctwo$, since the two source-to-relay links are independent. On the other hand, the decoder at the destination decodes both users jointly from the channel observations of both source transmissions plus the additional syndrome bits obtained from the relay. Fig.~\ref{fig:decoderStructure} shows the factor graph of the decoder (if one ignores the ``correlation'' part).


For source $\src_i$, we denote the messages (erasure probabilities) sent from a variable node at position $t$ in iteration $I$ to a check node in the first and the second layer as 
$\pit{i}{t}{I}$ and $\pits{i}{t}{I}$, 
respectively. Likewise, the messages from check nodes at position $t$ in iteration $I$ to variable nodes are called $\qit{i}{t}{I}$ and $\qits{i}{t}{I}$. For $t\notin [1,L]$, we have $\pit{i}{t}{I}=\pits{i}{t}{I}=0$. The DE update equations for source $\src_1$ for $t\in [1,L]$ are given as 
\begin{align}
\pit{1}{t}{I+1} & = \epsoned \bigg( \frac{1}{w} \sum_{j=0}^{w-1} \qit{1}{t+j}{I} \bigg)^{\!\lone-1}
                 \bigg( \frac{1}{w} \sum_{j=0}^{w-1} \qits{1}{t+j}{I}  \bigg)^{\!\lsone}, \nonumber\\
\qit{1}{t}{I} & = 1-\bigg( 1-\frac{1}{w} \sum_{k=0}^{w-1} \pit{1}{t-k}{I}\bigg)^{\!\rone-1}
\label{eq:qone}, \\
\pits{1}{t}{I+1} & = \epsoned \bigg( \frac{1}{w} \sum_{j=0}^{w-1} \qit{1}{t+j}{I} \bigg)^{\!\lone}
                 \bigg( \frac{1}{w} \sum_{j=0}^{w-1} \qits{1}{t+j}{I} \bigg)^{\!\lsone-1}, \nonumber\\
\qits{1}{t}{I} & = 1- \bigg( 1- \frac{1}{w} \sum_{k=0}^{w-1} \pits{1}{t-k}{I} \bigg)^{\!\rsone-1}  \nonumber\\
&~~~\cdot\bigg( 1- \frac{1}{w} \sum_{k=0}^{w-1}  \pits{2}{t-k}{I} \bigg)^{\!\rstwo}.
\label{eq:qsyndone}
\end{align}
The equations for source $\src_2$ are simply obtained by changing the source indices. The coupling of the codes of the two sources in the decoding process manifests itself in the messages \eqref{eq:qsyndone} sent from the second layer check nodes.
\begin{figure}[t]
\begin{center}
\includegraphics[scale=0.65]{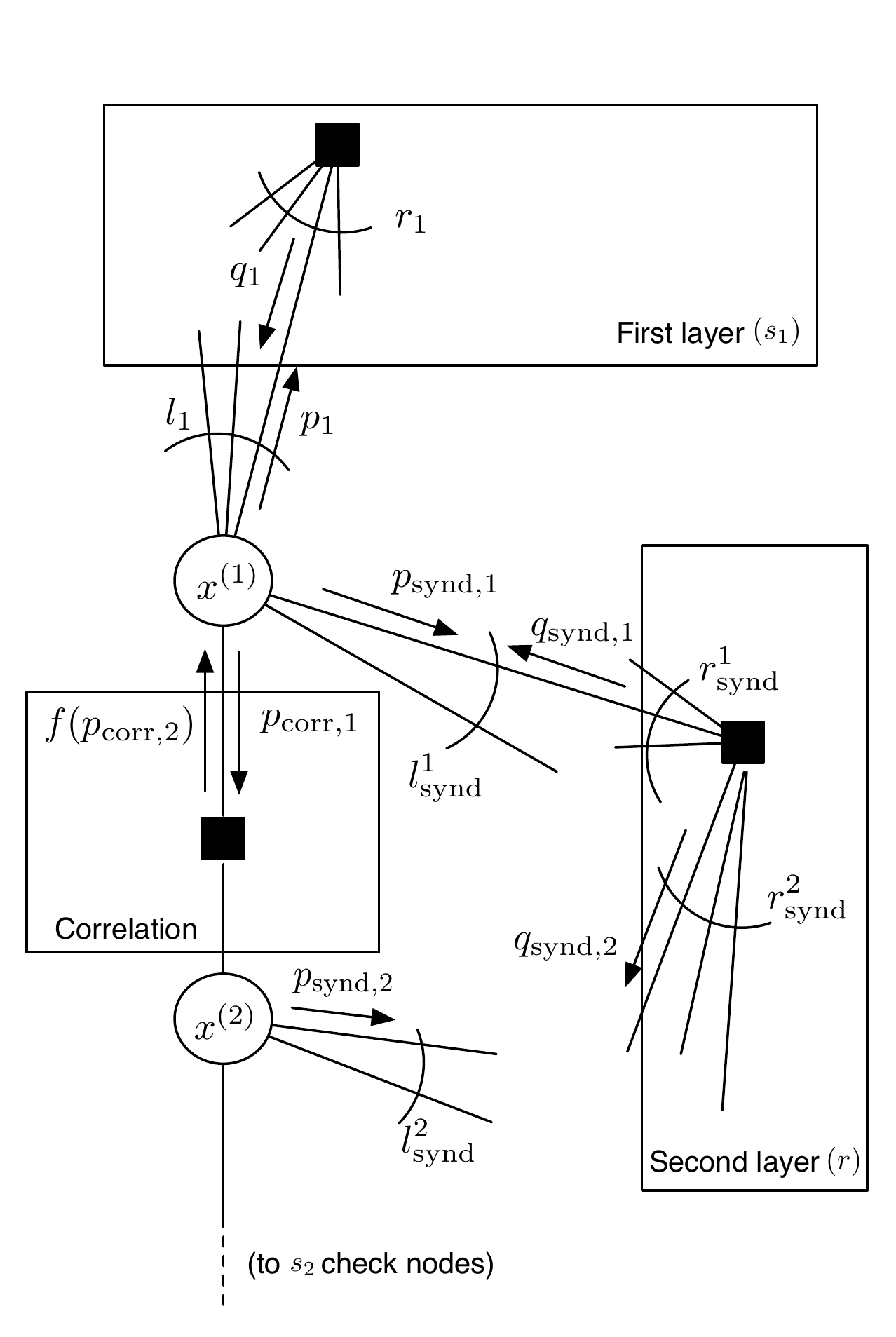}
\caption{Graph structure for the relay system with correlated sources. $x^{(1)}$ and $x^{(2)}$ shown here are systematic bits. The parity bits are connected to a channel factor node instead of a correlation factor node. The spatial coupling is not shown.}
\label{fig:decoderStructure}
\end{center}
\vspace{-3ex}
\end{figure}


\subsection{Density evolution for correlated sources}

The fact that the sources are correlated is important for the design of the decoders both at the relay as well as at the destination. The task of the relay is to detect two correlated sources whose data is transmitted over two independent noisy channels. A joint source-channel coding scheme for this scenario based on SC-LDPC codes was described in \cite{yedla2011universality}. We focus therefore now on the decoder at the destination, assuming the relay has successfully decoded the source data (which is the common assumption for DF relaying schemes).
The detector at the destination has to incorporate the syndrome bits obtained from the relay, in addition to taking the correlation between the sources into account. The structure of the resulting BP decoder (corresponding to the overall parity-check matrix $\H$ from \eqref{eq:overall_decoding_joint}) is depicted in Fig.~\ref{fig:decoderStructure}. Note that the variable nodes depicted in the figure correspond to systematic bits. For non-systematic bits, the factor node due to the correlation model does not exist; each non-systematic bit is connected to a factor node describing the channel law (erasure probability of the channel) instead.

For this decoder, the density evolution update equations for source $\src_1$ are
\begin{align}
\pit{1}{t}{I+1} & = \left[ \gamma_1 f\left( \pitc{2}{t}{I} \right) + (1-\gamma_1) \epsoned \right]\nonumber\\
&~~~\cdot          \bigg( \frac{1}{w} \sum_{j=0}^{w-1} \qit{1}{t+j}{I} \bigg)^{\!\lone-1}      \bigg( \frac{1}{w} \sum_{j=0}^{w-1} \qits{1}{t+j}{I}  \bigg)^{\!\lsone}, \nonumber\\
\qit{1}{t}{I} & = 1-\bigg( 1-\frac{1}{w} \sum_{k=0}^{w-1} \pit{1}{t-k}{I}\bigg)^{\!\rone-1},
\nonumber \\
\pits{1}{t}{I+1} & = \left[ \gamma_1 f\left( \pitc{2}{t}{I} \right) + (1-\gamma_1) \epsoned \right]\nonumber\\
&~~~\cdot\bigg( \frac{1}{w} \sum_{j=0}^{w-1} \qit{1}{t+j}{I} \bigg)^{\!\lone} \bigg( \frac{1}{w} \sum_{j=0}^{w-1} \qits{1}{t+j}{I} \bigg)^{\!\lsone-1}, \nonumber\\
\qits{1}{t}{I} & = 1- \bigg( 1- \frac{1}{w} \sum_{k=0}^{w-1} \pits{1}{t-k}{I} \bigg)^{\!\rsone-1}\nonumber\\
&~~~\cdot\bigg( 1- \frac{1}{w} \sum_{k=0}^{w-1}  \pits{2}{t-k}{I} \bigg)^{\!\rstwo}, \label{eq:qSyndUpdate} \\
\pitc{1}{t}{I+1} & = \bigg( \frac{1}{w} \sum_{j=0}^{w-1} \qit{1}{t+j}{I} \bigg)^{\!\lone}
                 \bigg( \frac{1}{w} \sum_{j=0}^{w-1} \qits{1}{t+j}{I}  \bigg)^{\!\lsone}.\nonumber
\end{align}

Equation \eqref{eq:qSyndUpdate} for the updates from the relay check nodes to the variable nodes of source $\src_1$ is valid under the assumption \eqref{eq:consistencyAss}. The fraction of systematic bits (which is equal to the unpunctured code rate) is denoted as $\gamma_1$. The equations for source $\src_2$ have the same structure, with switched source indices. The function of the correlation factor node is
$f(a) = (1-p)+ pa$, where
$p$ is the correlation parameter. This expresses that a factor node connecting the graphs of the two users exists with probability $p$, and there is no connection with probability $(1-p)$ (cf. correlation model from Section \ref{sec:sysmodel}).

\subsection{A proof for symmetric channels and uncorrelated sources}
\label{sec:ProofUncSources}

In the following, for the case of uncorrelated sources and symmetric channel conditions, defined as $\epsoner=\epstwor=\epsr$ and $\epsoned=\epstwod=\epsd$, we prove that the proposed scheme achieves the highest possible DF rate on the TD-MARC with BEC links. We call the two-user bilayer ensemble consisting of $\Ci(l,r,L,w)$ and $\Csi(\ls,\rs,L,w)$ the $\Cbl(l,\ls,r,\rs,L,w)$ ensemble (we consider the case $M\rightarrow\infty$ in the remainder of this section).
\begin{lemma}
For the case of uncorrelated sources and symmetric channel conditions
and $\rs=r/2$, the two-user bilayer code $\Cbl(l,\ls,r,r/2,L,w)$ achieves the same DE threshold for each 
source-destination link as the single-layer code $\Ccode(l+\ls,r,L,w)$.
\label{lemma:deThresh}
\end{lemma}
\begin{proof}
First note that the choice $\rs=r/2$ allows us to write the DE equations in a form that we need to prove the capacity-achieving property. However, this does not restrict the possible rates available in the system and therefore the result is general.

Assuming symmetric channel conditions, both source nodes use codes $\Ci$ from the same ensemble $(l,r,L,w)$, and the relay generates the syndrome bits using two codes $\Csi$ from the same ensemble $(\ls,\rs,L,w)$. This means $\lone=\ltwo=l$, $\rone=\rtwo=r$, $\lsone=\lstwo=\ls$ and $\rsone=\rstwo=\rs$.
The initial variable-to-check messages in the first iteration are equal for both sources and both layers, $\pit{1}{t}{1}=\pits{1}{t}{1}=\pit{2}{t}{1}=\pits{2}{t}{1}=\epsd$. With
$\rs=\frac{r}{2}$,
\eqref{eq:qsyndone} becomes equal to \eqref{eq:qone}, i.e., the check-to-variable messages of the two users in both layers are equal. This means that in the second iteration the variable-to-check messages will be equal again, and via induction, the same will happen in all the following iterations. Due to the assumed symmetry, the equations for the second user are the same as for the first user. The DE for each user can therefore be written as
\begin{align*}
\pit{i}{t}{I+1} 
&= \epsd \bigg( \frac{1}{w} \sum_{j=0}^{w-1} \qit{i}{t+j}{I} \bigg) ^{\!l+\ls-1}\\ 
&= \epsd \bigg(1-\! \frac{1}{w} \sum_{j=0}^{w-1} \bigg( 1\!-\! \frac{1}{w} \sum_{k=0}^{w-1} \pit{i}{t+j-k}{I}\bigg)^{\!r-1}  \bigg)^{\!l+\ls-1},
\end{align*}
which is the update equation for a single-layer SC-LDPC code ensemble $(l+\ls,r,L,w)$. The bilayer code ensembles for both users will therefore have the same DE thresholds as the single-layer ensemble.
\end{proof}

\begin{lemma}
For uncorrelated sources and symmetric channel conditions and $\rs=r/2$, the design rate of the two-user bilayer code $\Cbl(l,\ls,r,\rs,L,w)$ for each source-destination
link approaches that of the $\Ccode(l+\ls,r,L,w)$ single-layer ensemble
for $L,w\rightarrow\infty$.
\label{lemma:rate}
\end{lemma}
\begin{IEEEproof}
The design rate of an SC-LDPC code is 
$R=1-\frac{\nCheck}{\nVar}$ \cite{Kud11},
where $\nVar$ and $\nCheck$ denote the number of variable nodes and check nodes, respectively, in the graph. The number of variable nodes per user is $\nVar=ML$. The number of checks per user in the first layer is
\begin{equation*}
\nCheck = M \frac{l}{r}\bigg[ L + 1 + w - 2\sum_{j=0}^{w-1} \left( \frac{j}{w}\right)^{\!r} \bigg],
\end{equation*}
and the number of checks, shared by both users, in the second layer is
\begin{equation*}
\nCheckS = M \frac{2 \ls}{r}\bigg[ L + 1 + w - 2\sum_{j=0}^{w-1} \bigg( \frac{j}{w}\bigg)^{\!\frac{r}{2}}\bigg].
\end{equation*}
As the checks in the second layer are shared equally between the two users, the effective number of checks per user is $\nCheckSEff=\nCheckS/2$ (cf.~\eqref{eq:n_synd_bits}). The rate of the bilayer code for each user is therefore
$\Rbl=1-\frac{\nCheck+\nCheckS/2}{ML}$.
For $L,w\rightarrow\infty$, in that order,
\begin{equation*}
\lim_{w\rightarrow\infty}\lim_{L\rightarrow\infty} \Rbl = 1- \frac{l+l_\synd}{r},
\end{equation*}
which is the same expression as for the rate of a single-layer $(l+\ls,r,L,w)$ code ensemble in the limit $L,w\rightarrow\infty$.
\end{IEEEproof}
\begin{corollary}
For uncorrelated sources and symmetric channel conditions, the two-user bilayer SC-LDPC code $\Cbl(l,\ls,r,r/2,L,w)$ has the rate of the single-layer $(l+\ls,r,L,w)$ code,
\begin{equation*}
\lim_{w\rightarrow\infty}\lim_{L\rightarrow\infty}R(l,\ls,r,r/2,L,w)=1-\frac{l+\ls}{r}.
\end{equation*}
For fixed $(l+\ls)/r$, its BP threshold tends to the Shannon limit,
\begin{equation*}
\lim_{(l+\ls)\rightarrow\infty}\lim_{w\rightarrow\infty}\lim_{L\rightarrow\infty}\bpThresh(l,\ls,r,r/2,L,w)=\frac{l+\ls}{r}.
\end{equation*}
\end{corollary}
\begin{IEEEproof}
From Lemmas \ref{lemma:deThresh} and ~\ref{lemma:rate} we know that in the limit of large $w$ and $L$, the bilayer code ensemble has the same rate and DE threshold as the single-layer $(l+\ls,r,L,w)$ code ensemble.
Therefore, the corollary follows from \cite[Theorem 10 and Lemma 8]{Kud11}.
\end{IEEEproof}

\begin{theorem}
For a binary erasure TD-MARC with two uncorrelated sources, one relay, one destination, and for symmetric channel conditions, there exists an SC-LDPC code $\Ccode$ and an associated two-user bilayer code $\Cbl$ 
such that $\Ccode$ achieves the capacity for both source-relay
links and $\Cbl$ achieves capacity for both 
source-destination
links. In addition, this code construction achieves the highest possible rate with DF relaying.
\end{theorem}
\begin{IEEEproof}
Recall that the capacities of the 
source-relay links are
$\Csr=1-\epsr$.
We use capacity-achieving SC-LDPC codes from the ensemble $\Ccode(l,r,L,w)$, with
\begin{equation}
\frac{l}{r}=\epsr,\label{eq:epsrparams}
\end{equation} which are known to be asymptotically capacity achieving (cf.~\eqref{eq:capAch}), and therefore the relay will be able to decode successfully.

Let $\nVar$ be the number of variable nodes in $\Ccode$. In the limit $L\rightarrow\infty$ there are
$\nCheck=\frac{l}{r} \nVar$
check nodes. The effective number of additional bits needed by the destination and provided by the relay is (cf. \eqref{eq:n_synd_bits}, 
\eqref{mu1mu2}, $\mu_i=1/2$)
\begin{equation}
\nCheckSEff=\nCheckS/2=\nVar (\Csr-\Csd) = \nVar (\epsd-\epsr).
\label{eq:nEffChks}
\end{equation}
As $\rs=r/2$,
the additional effective $\nCheckSEff$ check nodes from the second layer add $\nCheckS r/2=r\nCheckSEff$ edges. The variable node degree $\ls$ is
\begin{equation}
\ls=r\nCheckSEff/\nVar = r(\epsd-\epsr).\label{eq:lsynd}
\end{equation}
From Corollary 1, together with \eqref{eq:epsrparams} and \eqref{eq:lsynd}, it follows that
\begin{equation*}
\lim_{w\rightarrow\infty}\lim_{L\rightarrow\infty}R(l,\ls,r,r/2,L,w)=1-\frac{l+\ls}{r}=1-\epsd,
\end{equation*}
and for fixed $(l+\ls)/r$,
\begin{equation*}
\lim_{(l+\ls)\rightarrow\infty}\lim_{w\rightarrow\infty}\lim_{L\rightarrow\infty}\bpThresh(l,\ls,r,r/2,L,w)=\epsd.
\end{equation*}
Therefore, we see that $\Cbl$ achieves the capacity of the source-destination links $\Csd=1-\epsd$ for both users. The number of channel uses in the first two transmission phases is $\none=\ntwo=\nVar$. A capacity-achieving SC-LDPC code is used to transmit the $\nCheckS$ syndrome bits in the third phase, using $\nr=\nCheckS/\Crd$ channel uses. Thereby we have shown (using \eqref{eq:nEffChks})
\begin{equation*}
\frac{\ni}{\none+\ntwo+\nr}=\frac{\Crd}{2\Crd+2(\Csr-\Csd)}=\theta_i^*,
\end{equation*}
i.e., our code design uses the optimum time allocation \eqref{eq:theta1opt}, which maximizes the achievable rate.
\end{IEEEproof}

While a proof for the general case of correlated sources and non-symmetric channel conditions is more difficult to obtain, in Section~\ref{sec:numresults} we use DE to demonstrate that the proposed code construction approaches capacity also for the general scenario.

\section{Achievable Channel Parameters for Given Code Rates}
\label{sec:achChanParam}

In this section, we derive the theoretical achievable region of channel parameters $(\epsoner,\epstwor)$ and $(\epsoned,\epstwod)$ for the s-r and s-d channels, respectively, corresponding to the pairs of channel parameters for which error-free transmission is possible, given the transmission rates $\Rone$, $\Rtwo$, $\Rsone$ and $\Rstwo$. This region is the theoretical benchmark to which the performance of the proposed two-user bilayer code ensembles must be compared. We carry out the derivations first for non-punctured codes (to be used for uncorrelated sources) and then for punctured codes (used for correlated sources).

We consider the $\src_i$-$\rel$ links first. From \eqref{eq:sr1x} and \eqref{eq:sr2x}, we obtain
\begin{align}
&R_1 \le \frac{1}{\Rsrcone}(1-\epsoner), \quad R_2 \le \frac{1}{\Rsrctwo}(1-\epstwor),\nonumber\\
&
\Rightarrow \epsoner \le 1 - \Rsrcone R_1 , \quad \epstwor \le 1 - \Rsrctwo R_2 .\label{eq:eps1reps2rlim}
\end{align}

\begin{figure*}
\centering
\subfloat[][]{
\includegraphics[scale=0.6]{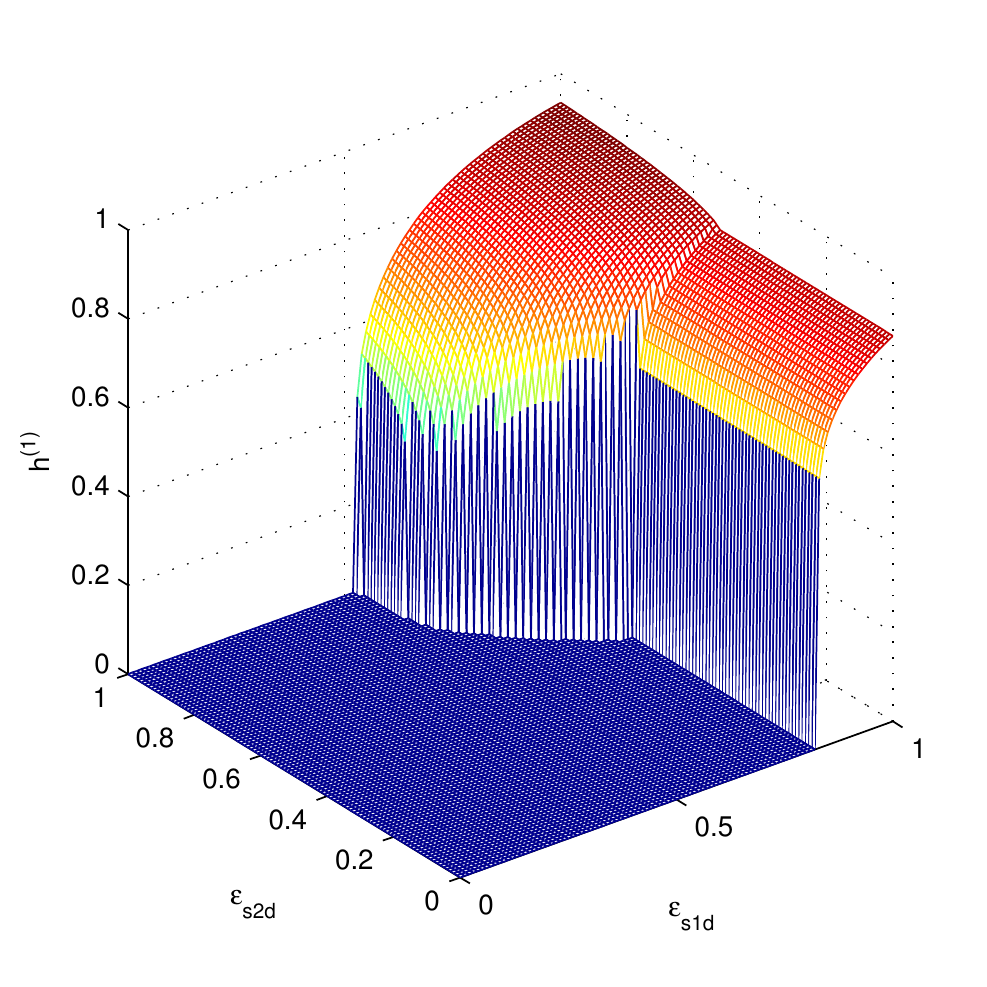}
}
\subfloat[][]{
\includegraphics[scale=0.6]{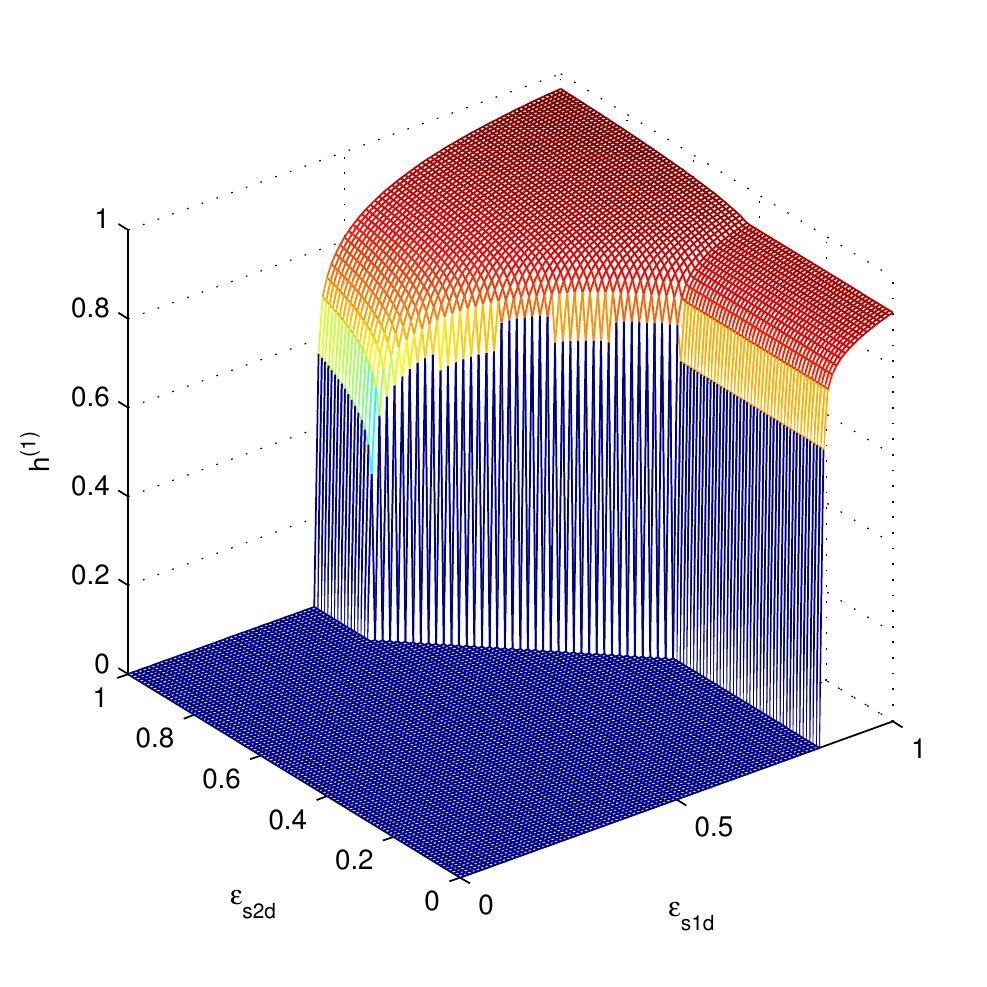}
}
\subfloat[][]{
\includegraphics[scale=0.6]{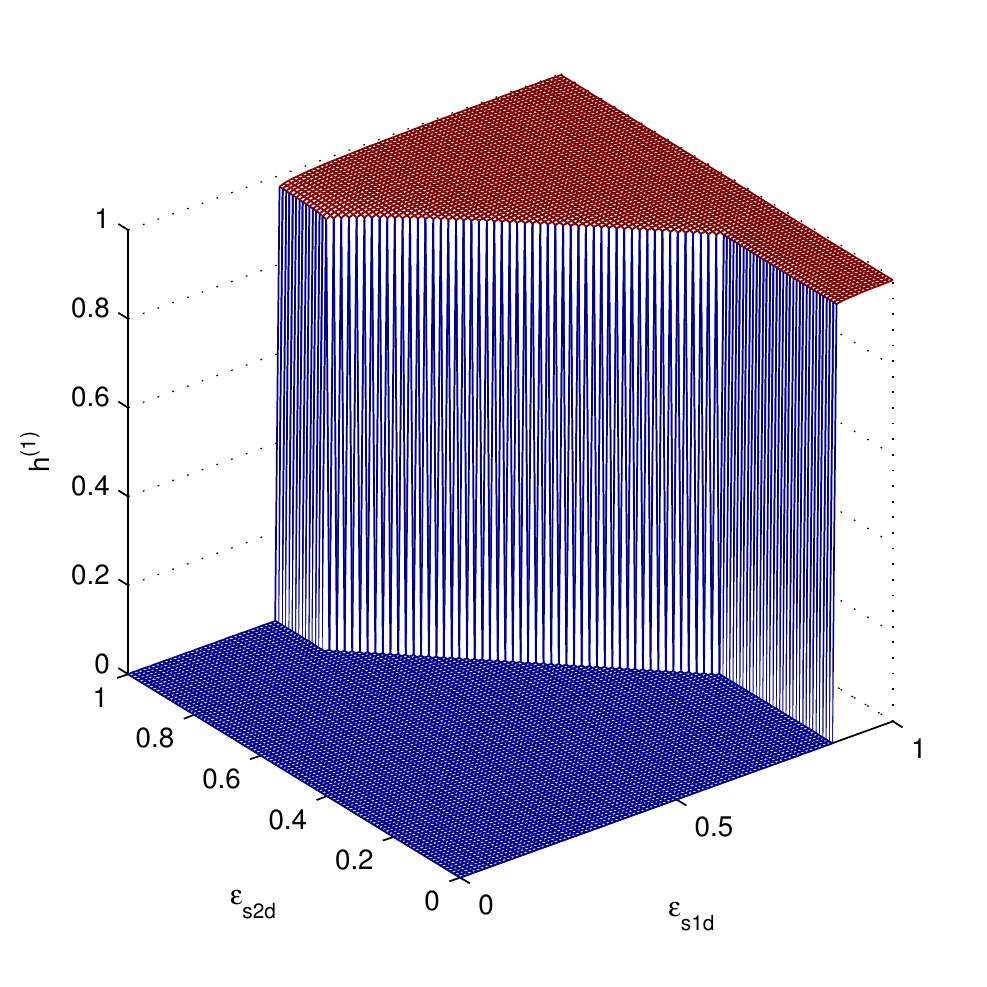}
}
\caption{3D EXIT chart $h^{(1)}(\epsoned,\epstwod)$ for code A, source $\src_1$, correlation parameter $p=0.3$. The chain-length is a) $L=10$ b) $L=20$ c) $L=300$.}
\label{fig:2dexitaa}
\vspace{-3ex}
\end{figure*}

The system with joint source-channel coding does not have specific source coding rates. We obtain the achievable $(\epsoner,\epstwor)$-region as the union of all regions defined by the inequalities \eqref{eq:eps1reps2rlim} for all $H(U_1|U_2)\le\Rsrcone\le1$. The resulting region is a pentagon with corner points $(0,0)$, $(a_{\src\rel},0)$, $(a_{\src\rel},b_{\src\rel})$, $(c_{\src\rel},d_{\src\rel})$, and $(0,d_{\src\rel})$, where
\begin{align*}
a_{\src\rel}& = 1-H(U_1|U_2) R_1, &\quad d_{\src\rel}&=1-H(U_2|U_1) R_2,\\
c_{\src\rel}&=1-R_1, &\quad b_{\src\rel}&=1-R_2.
\end{align*}

For the $\src_i$-$\dst$ links, we consider the inequalities \eqref{eq:srd1x}, \eqref{eq:srd2x} and $\eqref{eq:sumConstx}$. Using
\begin{equation*}
\frac{\theta_r}{\theta_i}=\frac{\nr}{n_i}=\frac{\kr}{n_i\Rr}=\frac{1}{\Rr}(1-\Rsone)
\end{equation*}
and $\Rr=\Crd$, we obtain the following expressions. From \eqref{eq:srd1x} we obtain
\begin{align}
\frac{\Rsrcone \Reff-\theta_r \Crd}{\theta_1}&\le 1-\epsoned \nonumber\\
\Rightarrow \epsoned &\le 1-\left(\Rsrcone R_1 - \frac{1-\Rsone}{\Rr} \Crd \right )\nonumber \\
&= 1-\left( \Rsrcone R_1 - ( 1-\Rsone) \right)\label{eq:achepsone}.
\end{align}
From \eqref{eq:srd2x} we get
\begin{equation}
\epstwod \le  1-\left(  \Rsrctwo R_2 - ( 1-\Rstwo) \right)\label{eq:achepstwo}.
\end{equation}
Finally, from \eqref{eq:sumConstx} we obtain
\begin{equation}
\epstwod \le 1 - \left( R_2 H(U_1,U_2)-\frac{R_2}{R_1}(1-\epsoned)-(1-\Rstwo) \right).\label{eq:achepstightsym}
\end{equation}
We obtain the achievable $(\epsoned,\epstwod)$-region as the intersection of the union of all regions defined by inequalities \eqref{eq:achepsone} and \eqref{eq:achepstwo} for $H(U_1|U_2)\le\Rsrcone\le 1$ with the region defined by \eqref{eq:achepstightsym} for $0\le\epsoned\le1$. It is again a pentagon with corner points $(0,0)$, $(a_{\src\dst},0)$, $(a_{\src\dst},b_{\src\dst})$, $(c_{\src\dst},d_{\src\dst})$, and $(0,d_{\src\dst})$, where
\begin{align*}
a_{\src\dst}&=1-(H(U_1|U_2) R_1 - (1-\Rsone)), \\
d_{\src\dst}&=1-(H(U_2|U_1) R_2 - (1-\Rstwo)),\\
c_{\src\dst}&=1-R_1,\\
b_{\src\dst}&=1-R_2.
\end{align*}

%
When systematic punctured codes are used, we obtain the following modified limits. For the s-r links, instead of \eqref{eq:eps1reps2rlim}, we have
\begin{equation*}
\epsoner \le 1 -  \Rsrcone \Ronep, \qquad \epstwor \le 1 - \Rsrctwo \Rtwop.
\end{equation*}
The resulting achievable rate region $(\epsoner,\epstwor)$ is a pentagon with corner points
\begin{align*}
a_{\src\rel}^p&=1-H(U_1|U_2) \Ronep,  &\quad d_{\src\rel}^p&=1-H(U_2|U_1) \Rtwop,\\
c_{\src\rel}^p&=1-\Ronep, &\quad b_{\src\rel}^p&=1-\Rtwop.
\end{align*}

\begin{table*}
\scriptsize
\caption{Code ensembles.}
\vspace{-2ex}
\begin{center}\begin{tabular}{ccccccccccccc}
\toprule
Label & $(\lone,\rone)$ & $(\ltwo,\rtwo)$ &$(\lsone,\rsone)$ & $(\lstwo,\rstwo)$ & $\mu_1$ & $\mu_2$ & $w$ & $L$ & $\Ronep$ & $\Rtwop$ & $\Rsone$ & $\Rstwo$ \\
\otoprule
Code A & $(6,10)$ & $(6,10)$ & $(2,10)$ & $(2,10)$ & $0.5$ & $0.5$ & $10$ & $600$ & $0.6446$ & $0.6446$ & $0.7973$ & 0.7973\\[0.5mm]
Code B & $(12,20)$ & $(14,20)$ & $(4,14)$ & $(3,14)$ & $0.5$ & $0.5$ & $10$ & $600$ & $0.6427$ & $0.4080$ & $0.7102$ & 0.7827\\[0.5mm]
\midrule
\bottomrule
\end{tabular} \end{center}
\label{tbl:codes}
\end{table*}

The inequalities for the s-d links are
\begin{align*}
\epsoned &\le 1- \frac{1}{1-R_1}\left(\Rsrcone R_1 - (1-\Rsone)  \right),\\
\epstwod &\le 1- \frac{1}{1-R_2}\left(\Rsrctwo R_2 - (1-\Rstwo)  \right),\\
\epstwod &\le 1- \frac{1}{1-R_2}\bigg(  H(U_1,U_2) R_2\\
&~~~- \frac{R_2}{R_1} (1-R_1)(1-\epsoned)-(1-\Rstwo) \bigg).
\end{align*}
The pentagon describing the achievable $(\epsoned,\epstwod)$-region has the corner points
\begin{align}
a_{\src\dst}^p&=1-\left(H(U_1|U_2) \Ronep - \frac{1}{1-R_1} (1-\Rsone)\right),\label{eq:achRegionPunct1} \\
d_{\src\dst}^p&=1-\left(H(U_2|U_1) \Rtwop - \frac{1}{1-R_2} (1-\Rstwo)\right),\label{eq:achRegionPunct2}\\
c_{\src\dst}^p&=1-\Ronep,\qquad b_{\src\dst}^p=1-\Rtwop.\label{eq:achRegionPunct3}
\end{align}
Examples for these achievable regions are given in the next section.


\section{Numerical results}
\label{sec:numresults}

The outstanding performance of SC-LDPC ensembles is due to the saturation of the BP threshold to the MAP threshold of the underlying block code ensemble. This phenomenon has been proven to occur in single-user SC-LDPC $(l,r,L,w)$ ensembles for the BEC \cite{Kud11}. As we showed in Section~\ref{sec:ProofUncSources}, the density update equations of the two-user bilayer SC-LDPC code ensemble reduce to the ones of a single-user code in the case of uncorrelated sources and symmetric channel conditions. Therefore, for this case, threshold saturation occurs. In the following, we use BP EXIT functions \cite{richardson2008modern} to give empirical evidence that threshold saturation does also occur for the non-symmetric scenario with correlated sources. Furthermore, we show how the joint decoding of both users makes the system more robust.

To assess the performance of the proposed relaying scheme, we construct three two-user bilayer SC-LDPC code ensembles. Their parameters are given in Table~\ref{tbl:codes}. For Code A, the same code ensembles are used for $\Cone$ and $\Ctwo$, and for $\Csone$ and $\Cstwo$, i.e., this bilayer code is designed for the symmetric case. On the other hand, code B is designed for a system with $\src_i$-$\rel$ links of different reliabilities, which is reflected by the use of codes $\Cone$ and $\Ctwo$ with different rates. Note that in some cases, the ensembles used at the relay to generate extra syndrome bits have variable node degree 2. However, this is not a problem, since these codes are never used alone, but within the bilayer structure.

We use a two-dimensional EXIT analysis to show that the BP threshold of the bilayer SC-LDPC code ensembles closely approaches the limits derived in Section~\ref{sec:achChanParam}. In our system, there are two users and two decoders. We concentrate on the decoder at the destination, since the one at the relay was already described in \cite{yedla2011universality}. 
For each source, the decoding performance depends on both channel erasure probabilities $\epsoned$ and $\epstwod$. The two-dimensional EXIT function for source $\src_i$ is
\begin{equation*}
h^{(i)}(\epsoned,\epstwod)\triangleq\frac{1}{L}\sum_{t=1}^L m^{(i)}_t,
\end{equation*}
where 
\begin{equation*}
m_t^{(i)}\triangleq \lim_{I\rightarrow\infty}\bigg( \frac{1}{w} \sum_{j=0}^{w-1} \qit{i}{t+j}{I} \bigg)^{\!l_i}
                 \bigg( \frac{1}{w} \sum_{j=0}^{w-1} \qits{i}{t+j}{I}  \bigg)^{\!\lsi}.
\end{equation*}
To obtain a single parameter for the whole code chain, we average over all messages at different time instants (cf. \cite{Kud11}).
\begin{figure*}
\centering
\subfloat[][]{
\includegraphics[scale=0.6]{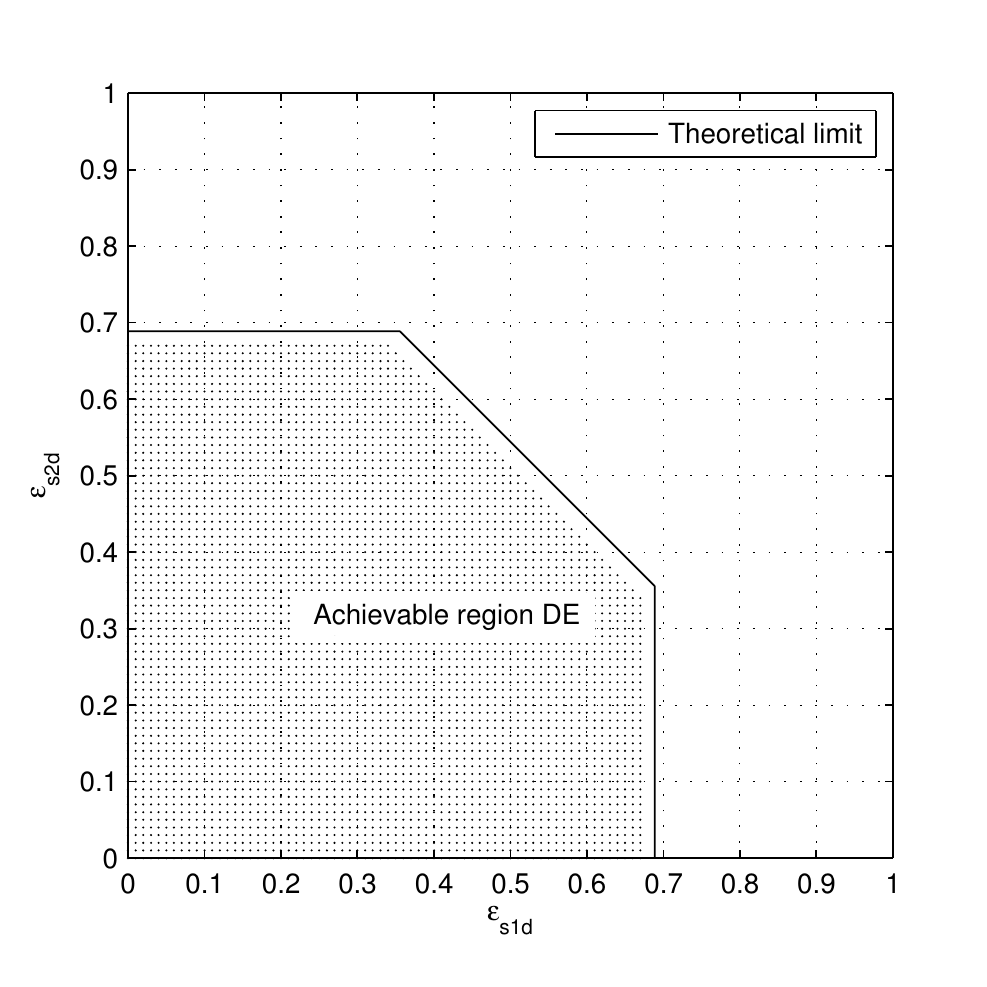}
}
\subfloat[][]{
\includegraphics[scale=0.6]{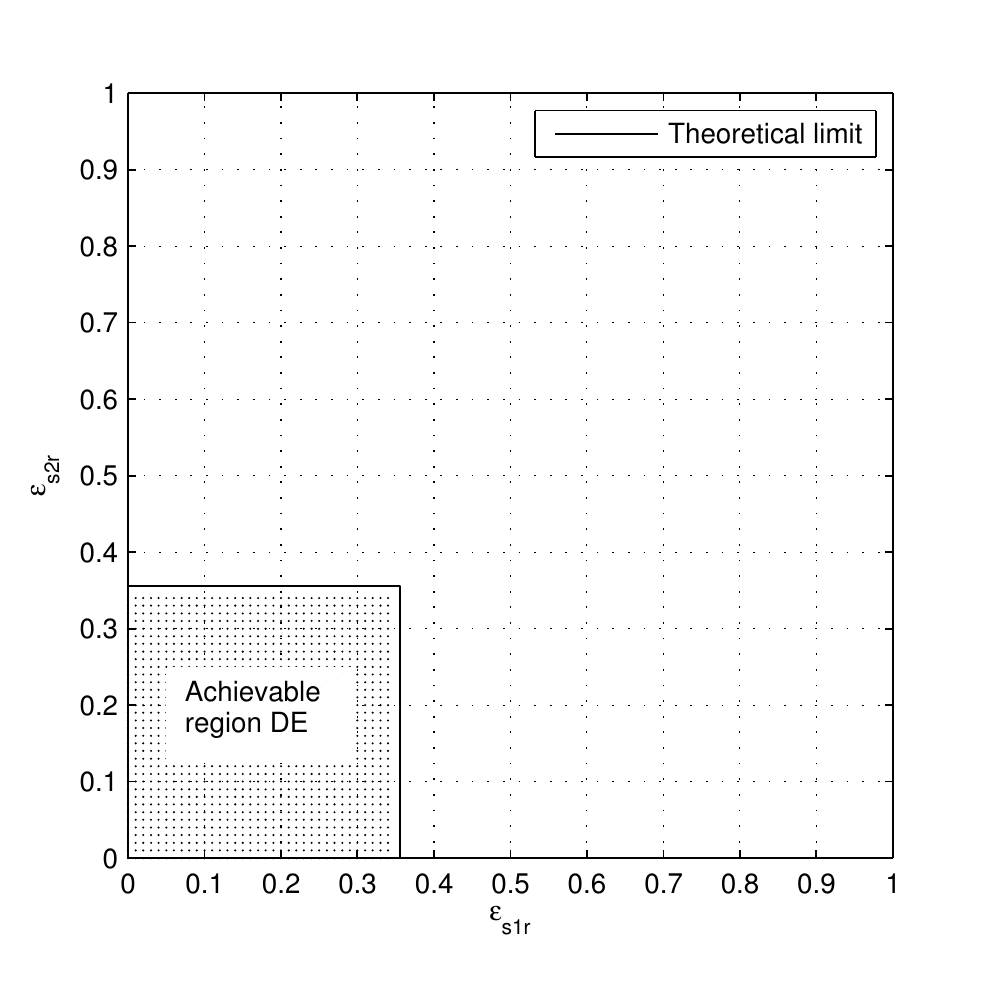}
}
\subfloat[][]{
\includegraphics[scale=0.6]{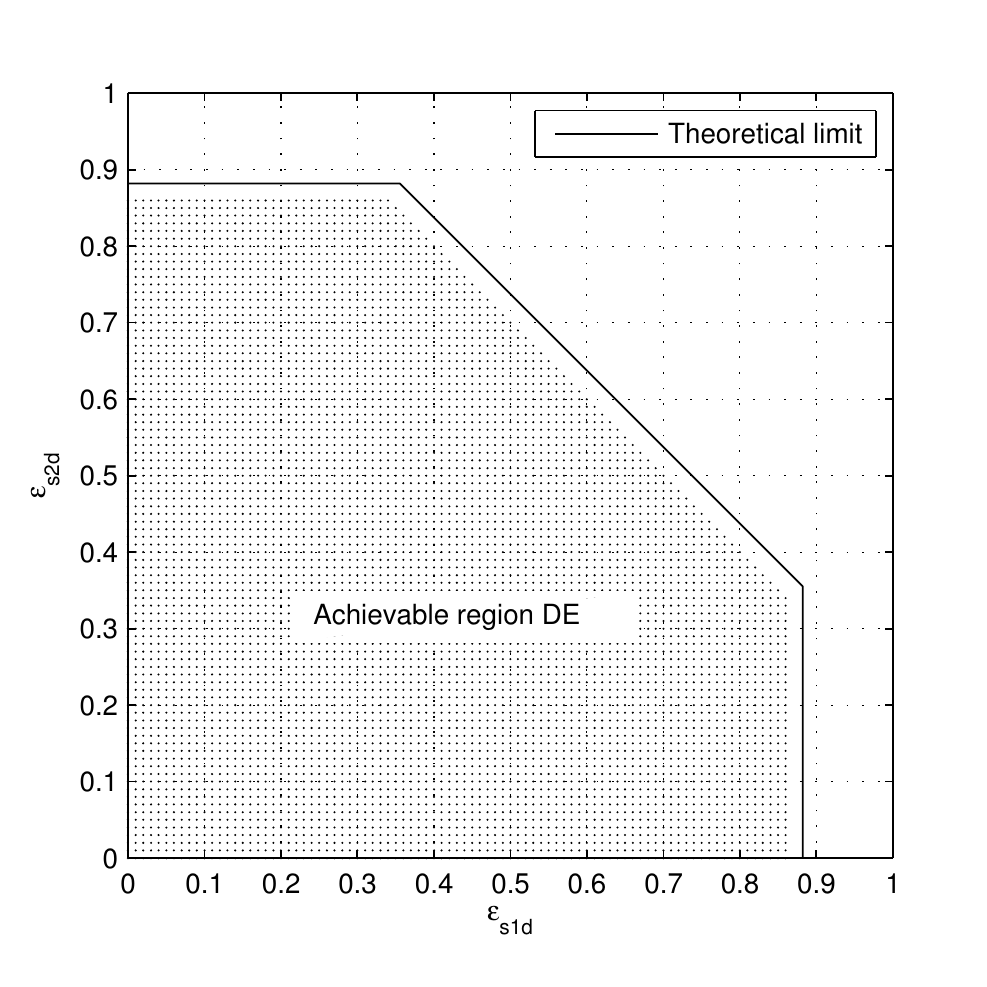}
\label{fig:achRegionCodeACfBer}
}
\vspace{-1ex}
\caption{a) Achievable $(\epsoned,\epstwod)$-region and theoretical limit for Code A, $p=0$ (no correlation). b) Achievable $(\epsoner,\epstwor)$-region and theoretical limit for Code A, $p=0$. c) Achievable $(\epsoned,\epstwod)$-region and theoretical limit for Code A, $p=0.3$.}
\label{fig:deCorrGen}
\vspace{-3ex}
\end{figure*}
\begin{figure*}
\subfloat[][]{
\vspace{-1ex}
\includegraphics[scale=0.6]{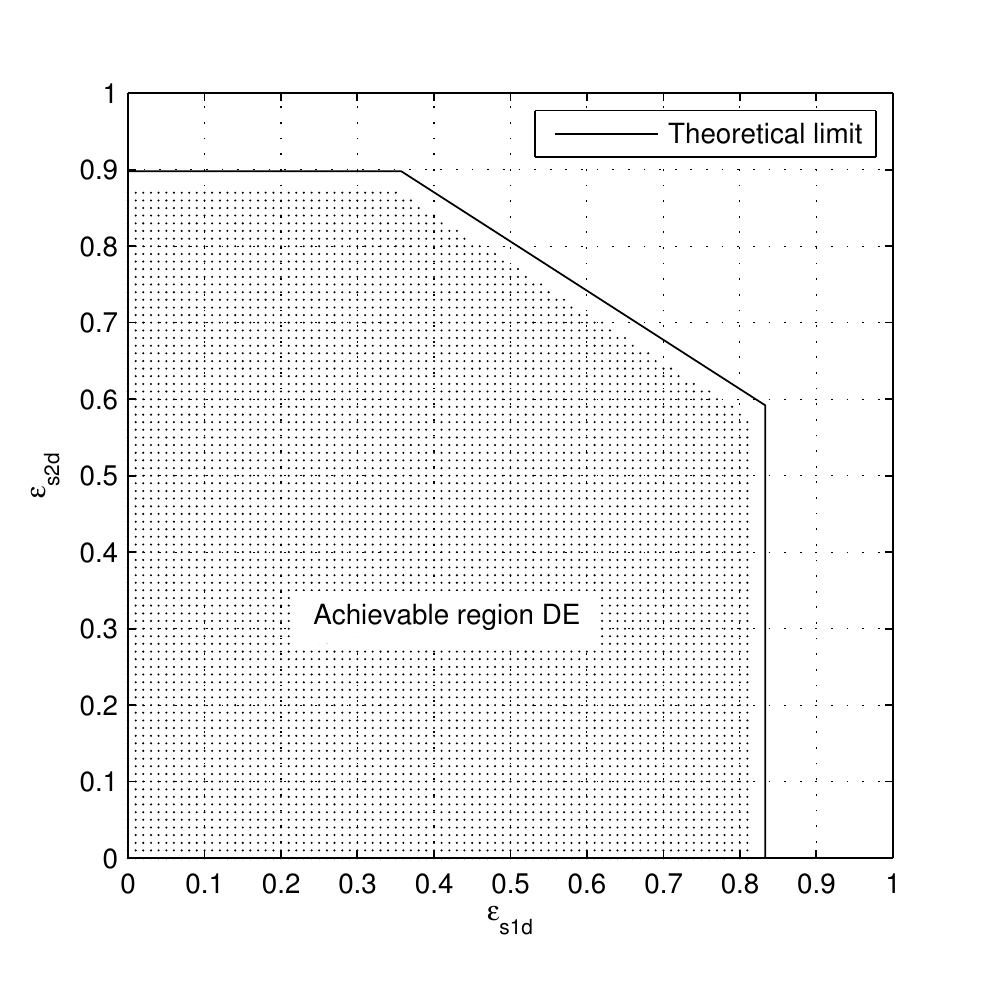}
}
\centering
\subfloat[][]{
\vspace{-1ex}
\includegraphics[scale=0.6]{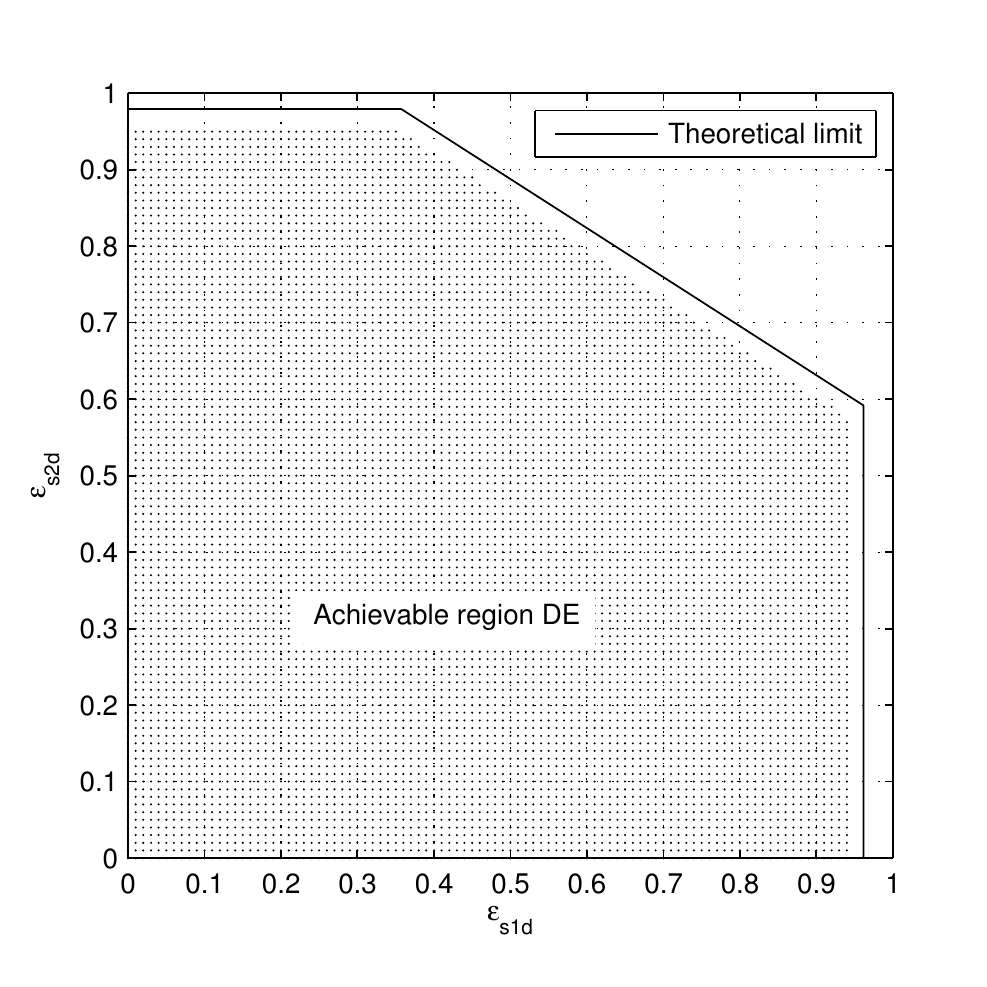}
}
\subfloat[][]{
\includegraphics[scale=0.6]{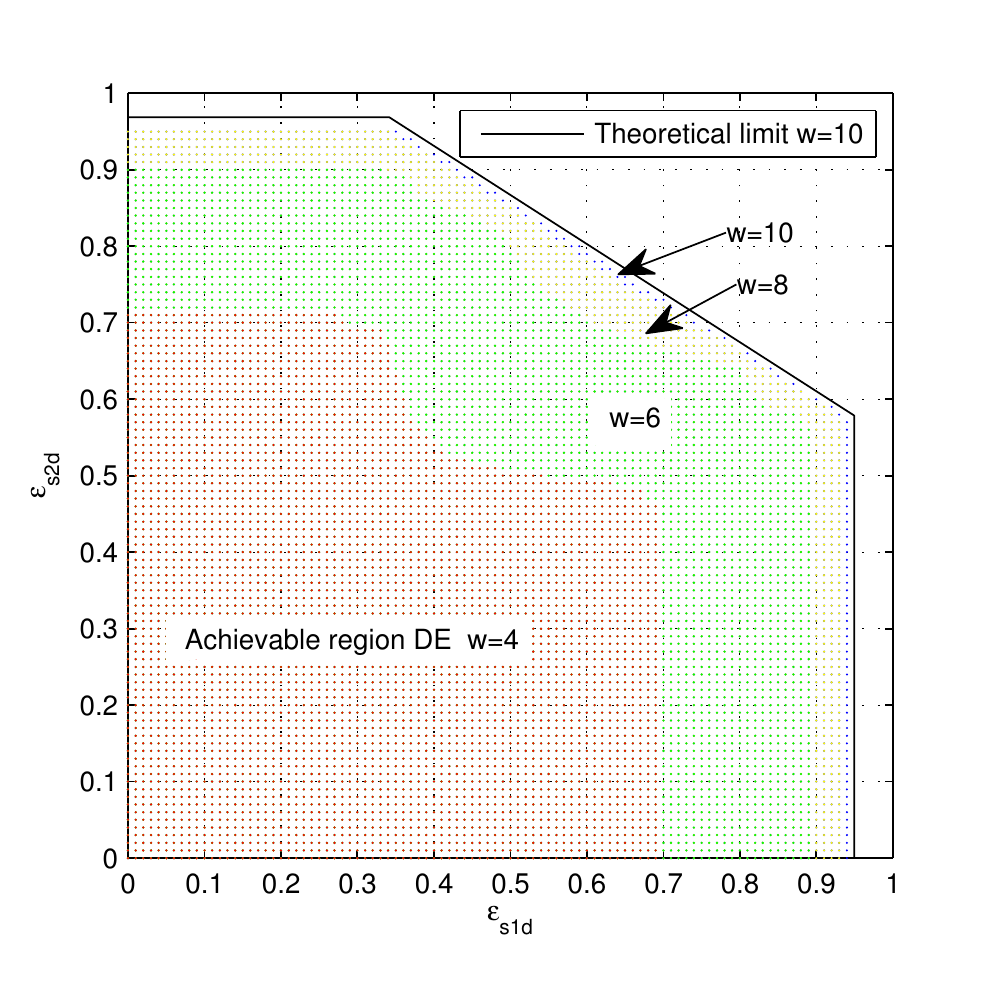}
}
\vspace{-1ex}
\caption{Achievable $(\epsoned,\epstwod)$-regions and theoretical limits for Code B. a) $p=0$, $w=10$. b) $p=0.2$, $w=10$. c) $p=0.2$, $w=4,6,8$ and $10$.}
\label{fig:deCorrGenC}
\vspace{-3ex}
\end{figure*}

Fig.~\ref{fig:2dexitaa}(a) shows the EXIT surface for source $\src_1$ for Code A (see Table~\ref{tbl:codes} for code parameters), correlation parameter $p=0.3$, and a chain length of $L=10$. The region where $h^{(1)}(\epsoned,\epstwod)\equiv0$ corresponds to the channel parameters $(\epsoned,\epstwod)$ for which error free decoding of $\src_1$ is possible. In the other region, $h^{(1)}(\epsoned,\epstwod)>0$. Fig.~\ref{fig:2dexitaa}(b) shows the same scenario for $L=20$. The transition from the error-free region to the region where $h^{(1)}$ is close to one is steeper. The threshold saturation phenomenon is already observed in Fig.~\ref{fig:2dexitaa}(c) for $L=300$. In this case the transition is very steep. Increasing the chain length further results in very small changes of the shape of the surface, denoting threshold saturation. While threshold saturation to the MAP threshold holds for the uncorrelated case with symmetric channel conditions, we remark that we cannot claim threshold saturation to the MAP threshold for the general case. However, the results in Fig.~\ref{fig:2dexitaa}(c), suggest that threshold saturation also holds for this case.

To compare the performance predicted by DE with the theoretical limits, we define the achievable $(\epsoned,\epstwod)$-region $\mathcal{A}$ as the region where both users can be decoded successfully at the destination,
\begin{equation*}
\mathcal{A}=\{(\epsoned,\epstwod)|h^{(1)}\equiv0 \text{ and } h^{(2)}\equiv0\}.
\end{equation*}

Fig.~\ref{fig:deCorrGen}(a) shows $\mathcal{A}$ (dotted area) for Code A, for the case where the sources are uncorrelated ($p=0$). The figure also depicts the theoretical limit given by \eqref{eq:achRegionPunct1}-\eqref{eq:achRegionPunct3} (black line). The proposed bilayer codes are near universal for the region $(\epsoned,\epstwod)$, i.e., they perform close to the theoretical limit for all channel conditions. We observe a largely uniform gap between the DE result and the theoretical limit of less than $0.02$. It can also be observed that a lower link quality for one user is compensated by a better quality on the other link. Note that this is possible because the data of both sources is combined via network coding at the relay and joint decoding at the destination. We can plot a similar figure for the achievable $(\epsoner,\epstwor)$-region (Fig.~~\ref{fig:deCorrGen}(b)), obtaining a similar behavior than for the region $\mathcal{A}$. Note that both sources can be decoded successfully both at the destination and at the relay for all pairs $(\epsoned,\epstwod)$ and $(\epsoner,\epstwor)$ within the dotted areas in Fig.~\ref{fig:deCorrGen}(a) and Fig.~\ref{fig:deCorrGen}(b), respectively.

In Fig.~\ref{fig:deCorrGen}(c) we plot $\mathcal{A}$ and the theoretical region for Code A and $p=0.3$. Compared to the uncorrelated case, the theoretical region as well as $\mathcal{A}$ become bigger. Now, the additional parity-checks due to the correlation model in the factor graph of the decoder allow successful decoding at higher channel erasure probabilities.
%


In Figs.~\ref{fig:deCorrGenC}(a) and~\ref{fig:deCorrGenC}(b), we plot $\mathcal{A}$ and the theoretical region for Code B for $p=0$ and $p=0.2$, respectively. As expected, both $\mathcal{A}$ and the theoretical region are now non-symmetric because of its dependency on $R_1$ and $R_2$. The DE results show again near universal performance. However, for both $p=0$ and $p=0.2$ we can observe a slightly larger gap to the theoretical limit on the side with the better limit ($y$-axis). When all parameters go to infinity, we would expect the whole theoretical region to be achieved.

We also analyzed the performance of the proposed bilayer codes as a function of the code parameters $w$, $l$ and $r$. In Fig.~\ref{fig:deCorrGenC}(c) we plot $\mathcal{A}$ for Code B and several values of $w$, $w=4,6,8$ and $10$, for $p=0.2$. (We remark that the theoretical region is slightly different for different values of $w$, as the resulting coding rates --which depend on $w$-- are slightly different. However, the four regions are very similar and, for the sake of clarity, in the figure we only plot the theoretical region corresponding to $w=10$). It is interesting to see that for $w=4$ and $6$ the achievable rate region is significantly smaller than the theoretical region. We recall that the theoretical limit is achieved by letting $w\rightarrow\infty$. In practice, for $w\geq 8$, near universal performance is obtained. A similar behavior is observed for the uncorrelated case. On the other hand, we observed that varying $l$ and $r$, while keeping constant $w$, leads to similar achievable regions. This seems to indicate that the impact of parameter $w$ on the performance is larger than that of the variable and check node degrees.

Finally, in Fig.~\ref{fig:berResults} we give bit erasure rates for Code A with $M=1800$, as a function of $\epsoned$. The correlation is $p=0.3$. The blue curve with triangles shows the performance for $\epsoned=\epstwod$, the black curve with circles for fixed $\epstwod=0.4$ and the red curve with squares for $\epstwod=0$. The corresponding DE thresholds and theoretical limits (cf. Fig.~\ref{fig:deCorrGen}(c)) are also given. We observe a performance close to the DE thresholds (note that the DE results require $M\rightarrow\infty$).

\section{Conclusions and Discussion}
\label{sec:conclusions}

In this work, we presented a bilayer spatially-coupled low-density parity-check code design for a system where two correlated sources transmit their data to a common destination with the help of a relay over BECs. The scheme combines joint source channel coding at the sources together with joint detection at both the relay and the destination to exploit the correlation. The relay employs network coding in order to enable joint decoding of both sources at the destination even in the non-correlated case, thereby improving the robustness of the system. We derived theoretical bounds on the achievable system rate, based on separate source and channel coding. Using density evolution, we showed that the proposed scheme approaches the theoretical limit for the general case of correlated sources and non-symmetric channel conditions. For the particular case of uncorrelated sources and symmetric channel conditions we proved analytically that our code construction achieves the heoretical limit. 

The use of BECs in this paper is motivated by the fact that it simplifies the DE analysis and allows to derive relevant proofs. In a more practical scenario, one should consider other channels, such as the Gaussian MARC. While most of the literature on SC-LDPC codes, and in particular those focused on its analysis, consider the BEC, we remark that it has been shown that SC-LDPC codes also perform well for other channels. For instance, for the 2-user Gaussian MAC channel, SC-LDPC codes have been shown to approach capacity \cite{yedla2011universality}. Therefore, we expect the codes proposed in this paper to perform close to the theoretical limit for the Gaussian MARC.

\begin{figure}
\begin{center}
\includegraphics[scale=0.7]{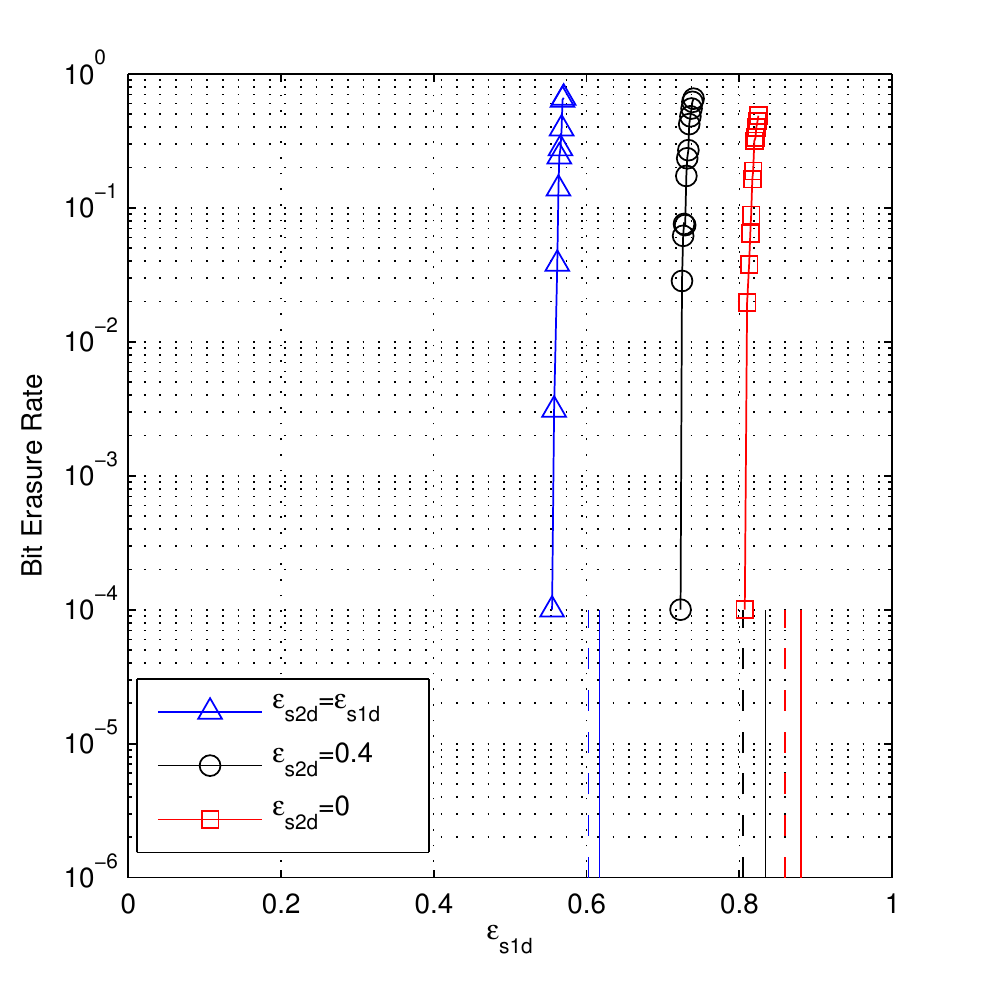}
\caption{Bit erasure rates (solid curves with markers), DE thresholds (dashed lines) and theoretical limits (solid lines)  for Code A, $p=0.3$.}
\label{fig:berResults}
\end{center}
\end{figure}
An alternative to SC-LDPC codes would be use of irregular LDPC codes. In \cite{yedla2010ldpc}, irregular LDPC codes are optimized for the transmission of two correlated sources over BECs. It is shown that, while performance close to the theoretical limit can be achieved for a range of channel parameters, for balanced channel parameters (i.e., similar erasure probabilities in the links) the codes perform away from the theoretical limit. As a result, the codes are not \textit{universal}. A similar observation was made in \cite{martalo2010density}. Note that transmission from the sources to the relay corresponds to the scenario in \cite{martalo2010density,yedla2010ldpc}. Therefore, it is reasonable to conjecture that irregular LDPC codes are not universal for the (more general) relaying scenario considered in this paper. In contrast, our DE results show that bilayer SC-LDPC code ensembles are near universal for the two-source relay channel (i.e., they perform close to the theoretical limit for the whole regions $(\epsoned,\epstwod)$ and $(\epsoner,\epstwor)$). On the other hand, we would like to remark that irregular LDPC codes require a degree distribution optimization using, e.g., differential evolution, which is computationally expensive. In the scenario considered in this paper, changing the channel conditions (i.e., the erasure probabilities $\epsoner$, $\epstwor$, $\epsoned$, and $\epstwod$) would not only have and impact on the code rates of each level, but it also requires a complete degree distribution optimization. This is not required for SC-LDPC codes, which makes the code design much simpler.

\subsection{Extension to $n$ sources}
The proposed bilayer SC-LDPC code can be extended to the $n$-source scenario. In this case, we can construct a bilayer code with parity check matrix
\begin{equation*}
\mathbf{H} = \left[\begin{array}{cccc}\mathbf{H}_1 &   \mathbf{0} &   \cdots &   \mathbf{0}\\  \mathbf{0}&\mathbf{H}_2 &  \cdots &\mathbf{0}\\
\vdots  &\mathbf{0} &  \ddots & \mathbf{0}\\
\mathbf{0}  &\cdots &  \mathbf{0} &\mathbf{H}_n\\
\Hsynd^1 &   \Hsynd^2 & \cdots &   \Hsynd^n \\
\H_\corr^1 &  \H_\corr^2 & \cdots &  \H_\corr^n \end{array}\right].
\label{eq:nsource}
\end{equation*}

For the uncorrelated case, the theoretical limits for the $n$-source scenario can be found in \cite{youssef2011distributed} and the code rates should be optimized according to these. We would like to remark that for the uncorrelated case and symmetric channel conditions (in this case we can set $l_i=l$, $r_i=r$, $l_{\synd}^i=l_{\synd}$ and $r_\synd^i=r_\synd$, for $i=1,\ldots,n$), as for the two-source case, the DE equation for each user (with $r_\synd=r/n$) can be written again as the DE equation for a single layer SC-LDPC code ensemble, and the proof in Section~\ref{sec:ProofUncSources} can be easily extended to the $n$-source case.


\appendix[Proof of the achievable DF rate for the TD-MARC with correlated sources]
\label{sec:appendixDfRate}

To obtain the maximum $\Reff$, we have to optimize the time allocation given by $\theta_1$, $\theta_2$ and $\theta_r$, and the source coding rates $\Rsrcone$ and $\Rsrctwo$. For this, we have to solve the following optimization problem 
\begin{align*}
\text{maximize} \quad &\min \{(f_1,\ldots,f_5)(\theta_1,\theta_2,\Rsrcone)\} \\
\text{s.t.}\quad & H(U_1|U_2) \le \Rsrcone \le H(U_1), \\
           & \Rsrctwo=H(U_1,U_2)-\Rsrcone, \\
            & \theta_r=1-\theta_1-\theta_2, \\ 
		    & 0 \le \theta_1,\theta_2 \le 1, \\
		    & \Csir \ge \Csid, \\
		    & \Crd  \ge \Csid.
\end{align*}
The functions $f_1,\ldots, f_5$ are defined in \eqref{eq:sr1x}-\eqref{eq:sumConstx}. Let $\theta_1^*$, $\theta_2^*$ and $\Rsrcone^*$ be the values for which $R'$ is maximized. To solve this optimization problem, we assume that out of the five inequalities \eqref{eq:sr1x}-\eqref{eq:sumConstx}, only three are active at the optimum $(\theta_1^*,\theta_2^*,\Rsrcone^*)$, namely \eqref{eq:sr1x}, \eqref{eq:sr2x} and \eqref{eq:sumConstx}. The optimum corresponds to the point where the three hyperplanes described by the three inequalities (taken as equalities) have their intersection. Then we will show that the other two constraints are not active at the optimum.

First we set $f_1(\theta_1,\Rsrcone)=f_5(\theta_1,\theta_2)$. This gives the relation
\begin{equation}
\theta_2^* = \theta_1^* \frac{\Csoner}{\Cstwor}\left(\frac{H(U_1,U_2)}{\Rsrcone}-1\right) = \kappa'(\Rsrcone)\theta_1^*\label{eq:theta_2_theta_1},
\end{equation}
where
$\kappa'=\kappa\left(\frac{H(U_1,U_2)}{\Rsrcone}-1\right)$. 
Then we set $f_2(\theta_2,\Rsrcone)=f_5(\theta_1,\theta_2)$ and insert \eqref{eq:theta_2_theta_1} for $\theta_2$, obtaining
\begin{equation}
\theta_1^*(\Rsrcone) = \frac{\Crd}{\left(1+\kappa'\right) \Crd + \frac{H(U_1,U_2)}{\Rsrcone}\Csoner -\Csoned - \kappa'\Cstwod}\label{eq:theta1_opt}.
\end{equation}


Using \eqref{eq:sr1x} the resulting achievable rate is
\begin{equation}
\label{eq:achRatea}
R'(\Rsrcone) = \frac{1}{\Rsrcone}\theta_1^* \Csoner.
\end{equation}
We still have to maximize this expression over $H(U_1|U_2)\le\Rsrcone\le 1$ to obtain the maximal achievable rate. The behavior of this function is determined by the value of $\alpha=(1-\kappa)\Crd-\Csoned+\kappa\Cstwod$. For $\alpha=0$, i.e., $\frac{\Csoner}{\Cstwor}=\frac{\Crd-\Csoned}{\Crd-\Cstwod}$, (\ref{eq:achRatea}) does not depend on $\Rsrcone$. For $\alpha>0$, it is monotonically decreasing in $\Rsrcone$ and therefore the maximum $\Reff$ is achieved for
\begin{equation*}
\Rsrcone^*=H(U_1|U_2), \qquad \Rsrctwo^*=H(U_2) = 1. 
\end{equation*}
For $\alpha<0$, it is monotonically increasing, and therefore the optimum choice is $\Rsrcone^*=1$ and $\Rsrctwo^*=H(U_2|U_1)$. This result tells us that, in general, one source should be maximally compressed while the other should not be compressed at all. Only in  the special case $\alpha=0$ the choice of $\Rsrcone$ can be made arbitrarily within the constraint 
\eqref{eq:rs2ofrs1} without affecting the effective transmission rate. 

We still have to show that inequalities \eqref{eq:srd1x} and \eqref{eq:srd2x} are fulfilled at the optimum $(\theta_1^*$, $\theta_2^*$ and $\Rsrcone^*$). First we show that the right-hand side of \eqref{eq:sr1x} is smaller than the right-hand side of \eqref{eq:srd1x}, i.e.,  $f_1(\theta_1^*,\Rsrcone^*)\le f_3(\theta_1^*,\theta_2^*,\Rsrcone^*)$. From $f_1(\theta_1^*,\Rsrcone^*)=f_5(\theta_1^*,\theta_2^*)$
we obtain
\begin{equation}
(1-\theta_1^*-\theta_2^*)\Crd=\frac{H(U_1,U_2)}{\Rsrcone}\theta_1^*\Csoner-\theta_1^*\Csoned-\theta_2^*\Cstwod.\label{eq:theta_r_opt}
\end{equation}
Inserting this into \eqref{eq:srd1x} gives
\begin{equation*}
\frac{\theta_1^*\Csoner}{\Rsrcone^*}\left( \frac{H(U_1,U_2)}{\Rsrcone^*}-\left(\frac{H(U_1,U_2)}{\Rsrcone^*}-1\right)\frac{\Cstwod}{\Cstwor}\right)\ge  \frac{\theta_1^*\Csoner}{\Rsrcone^*},
\end{equation*}
since $\frac{H(U_1,U_2)}{\Rsrcone}\ge  1$ and $\Cstwor\ge\Cstwod$. 
We can show in a similar way that inequality \eqref{eq:srd2x} is loose by comparing it to \eqref{eq:sr2x}.

\balance

\end{document}